\documentclass[12pt]{article}
\usepackage[utf8]{inputenc}
\usepackage{array,fullpage,multirow}
\usepackage{amssymb,amsmath,amsthm,sectsty,url,nicefrac,color}

\usepackage{paper}
\usepackage{subcaption}
\newcounter{num}
\makeatletter

\newcommand{\Daniel}[1]{\textcolor{blue}{[Daniel\@ifnotempty{#1}{: #1}]}}
\newcommand{\Zongchen}[1]{\textcolor{teal}{[Zongchen\@ifnotempty{#1}{: #1}]}}

\newcommand{\aw}[1]{\textcolor{orange}{[Alex\@ifnotempty{#1}{: #1}]}}

\makeatother

\usepackage{algorithm}
\usepackage{algpseudocode}
\algnewcommand\algorithmicinput{\textbf{Input:}}
\algnewcommand\algorithmicoutput{\textbf{Output:}}
\algnewcommand\Input{\item[\algorithmicinput]}%
\algnewcommand\Output{\item[\algorithmicoutput]}%
\usepackage{wrapfig}

\newcommand{\PP}{\mathbb{P}}
\newcommand{\EE}{\mathbb{E}}

\usepackage{verbatim}
\newcommand{\cA}{\mathcal{A}}
\newcommand{\cB}{\mathcal{B}}

\title{Time Lower Bounds for the Metropolis Process and Simulated Annealing}

\author{
Zongchen Chen\thanks{Department of Computer Science and Engineering, University at Buffalo. Email: \texttt{zchen83@buffalo.edu}}
\and
Dan Mikulincer\thanks{Department of Mathematics, MIT. Email: \texttt{danmiku@mit.edu.}}
\and 
Daniel Reichman\thanks{Department of Computer Science, WPI. Email: \texttt{dreichman@wpi.edu}}
\and 
Alexander S.\ Wein\thanks{Department of Mathematics, UC Davis. Email: \texttt{aswein@ucdavis.edu}. Partially supported by an Alfred P.\ Sloan Research Fellowship.}
}

\begin{document}

\maketitle

\begin{abstract}
    The Metropolis process (MP) and Simulated Annealing (SA) are stochastic local search heuristics that are often used in solving combinatorial optimization problems. Despite significant interest, there are very few theoretical results regarding the quality of approximation obtained by MP and SA (with polynomially many iterations) for NP-hard optimization problems.    

   We provide rigorous lower bounds for MP and SA with respect to the classical maximum independent set problem when the algorithms are initialized from the empty set. We establish the existence of a family of graphs for which both MP and SA fail to find approximate solutions in polynomial time. 
    More specifically, we show that for any $\varepsilon \in (0,1)$ there are $n$-vertex graphs for which the probability SA (when limited to polynomially many iterations) will approximate the optimal solution within ratio $\Omega\left(\frac{1}{n^{1-\varepsilon}}\right)$ is exponentially small. Our lower bounds extend to graphs of constant average degree $d$, illustrating the failure of MP to achieve an approximation ratio of $\Omega\left(\frac{\log (d)}{d}\right)$ in polynomial time. In some cases, our impossibility results also go beyond Simulated Annealing and apply even when the temperature is chosen adaptively.
    Finally, we prove time lower bounds when the inputs to these algorithms are bipartite graphs, and even trees, which are known to admit polynomial-time algorithms for the independent set problem. 
     
\end{abstract}

\thispagestyle{empty}

\newpage
\tableofcontents

\thispagestyle{empty}

\newpage

\setcounter{page}{1}

\section{Introduction}

Simulated Annealing~\cite{kirkpatrick1983optimization,vcerny1985thermodynamical} (SA) is a family of randomized local search heuristics, that is widely applicable for approximating solutions of combinatorial optimization problems. In the maximization version 
we are given a finite search space $\mathcal{C}$ of feasible solutions, and a cost $f(x)$, for every solution $x \in \mathcal{C}$. Additionally, for every solution $x \in \mathcal{C}$ there is a set $N(x) \subseteq \mathcal{C}$ of \emph{neighboring solutions} accessible via a single move of the search algorithm. In contrast to hill-climbing methods that consistently choose an element $y \in N(x)$ with  $f(y) \geq f(x)$, SA may choose a neighboring solution $y$ satisfying $f(y) < f(x)$ with probability $e^{-\frac{\Delta}{T}}$ where $\Delta:=f(x)-f(y)$ and $T>0$ is a \emph{temperature parameter}, that governs the behavior of the algorithm. Typically one gradually reduces the temperature over time allowing for a more exploratory algorithm in the early stages.   
The idea is that since the algorithm is allowed to accept downhill moves it should be able to escape local maxima. So, with an appropriately chosen cooling schedule, the hope is for the algorithm to find near-optimal solutions. The case where the temperature $T$ is fixed throughout the algorithm has received attention as well~\cite{metropolis1953equation}: in this case, the algorithm is called the Metropolis process (MP).

Since its inception in the 1980s SA was found empirically to be highly effective for numerous optimization problems in diverse fields such as VLSI design, pattern recognition, and quantum computing. The great popularity of SA is acknowledged in several dedicated books, articles, surveys, and textbooks concerned with algorithm design~\cite{kleinberg2006algorithm,dasgupta2008algorithms,aarts1989simulated,bertsimas1993simulated,johnson1989optimization,johnson1991optimization}. 

Owing to the wide applicability of SA for NP-hard optimization problems, one may wonder what rigorous results can be obtained regarding the algorithm's performance.
It is well-known~\cite{hajek1988cooling} that for a suitable cooling schedule, SA, if run for sufficiently many iterations, will almost surely converge to a global optimizer. However, there is no guarantee that the running time will be polynomial in the size of the input. This begs the question of what can be said with respect to SA when it is constrained to run for \emph{polynomially many steps}. This question is explicitly mentioned as challenging in several papers~\cite{aldous1994go,bertsimas1993simulated,gharan2011submodular,jerrum1996markov}. For example, it is mentioned in~\cite{aldous1994go} that ``The polynomial time behavior of simulated annealing is notoriously hard to analyze rigorously". In the field of approximation algorithms for NP-hard optimization problems not much seems to be known with respect to upper and lower bounds regarding the \emph{approximation factor} that can be achieved efficiently with SA. The situation for MP is similar: Little is known about the approximation ratio achievable by MP (when run for polynomially many steps) for NP-hard optimization problems. As stated by~\cite{jerrum1996markov}, ``Rigorous results...about the performance of the Metropolis algorithm on non-trivial optimization problems are few and far between". Despite some recent developments~\cite{coja2015independent,chen2022almost}, the literature on rigorous results for MP and SA remains sparse and experts have noticed the ``gap between theory and practice...for Simulated Annealing"~\cite{doerr2022simulated}.

The lack of runtime complexity lower bounds for SA and MP is not a coincidence, since some natural approaches run into difficulties. One direction to prove time lower bounds for MP and SA is to rely on known bounds for the mixing time of the relevant Markov chains. For such bounds, there is a wealth of existing established techniques~\cite{levin2017markov,mihail1989conductance,fill1991eigenvalue,chen1999lifting,dyer2002counting,bhatnagar2004torpid,mossel2009hardness}. However, slow mixing does not necessarily imply anything about the efficiency of MP as an approximation algorithm. For example, a simple conductance argument shows that MP has exponential mixing time, with any temperature parameter, when searching for the maximum independent set in the complete bipartite graph $K_{n,n}$. On the other hand, it is easily seen that with an appropriate temperature, MP will find an optimal independent set in $K_{n,n}$ in polynomial time. Furthermore, lower bounds on mixing times imply the existence of a ``bad" initial state from which the expected time of the chain to mix is super-polynomial. These kinds of statements do not usually carry information about initialization from specific states, as done in practice. For example, as observed in~\cite{jerrum1992large} and further elaborated in~\cite{chen2022almost}, conductance lower bounds on the mixing time do not imply comparable lower bounds on the time it takes MP for the independent set problem to find an optimal or even near-optimal solution when using the natural empty state initialization. Finally, as noted in~\cite{jerrum1992large}, common techniques to prove mixing times lower bounds for homogeneous Markov chains do not generalize in a straightforward way to inhomogeneous chains such as SA. 

The overarching aim of this paper is to further our understanding of the theoretical guarantees afforded by the above-mentioned algorithms and investigate the possible limitations and hard instances. Specifically, we focus on the \emph{maximum independent set} problem and establish lower bounds on the number of iterations required by MP or SA to obtain a reasonable approximation. As we detail in the next section, we consider several different families of graphs: dense graphs, graphs of bounded average degree, bipartite graphs, and trees. For each such class, we prove corresponding lower bounds, that differ in the obtained approximation ratio and allowed cooling schedule, used in SA. Notably, for dense graphs, which represent the most general category, we present particularly robust results that extend beyond SA, applying to \emph{any} cooling schedule, even adaptive ones.  While the specifics of our results, and their proofs, differ across graph classes, the common thread is this: We establish the existence of graph families where either MP or SA must run for an exponentially large number of steps to approximate the maximum independent set. The only exception is our lower bound for trees where we study the time to find the \emph{optimal} solution (as we will see in this case that MP does succeed at finding efficiently an \emph{approximate} solution for tree instances).

It is well known that the independent set problem is NP-hard not only to compute exactly but also to approximate~\cite{hastad1996clique,zuckerman2006linear,bhangale2022ug}. Nevertheless, the quest for \emph{unconditional} lower bounds for achieving efficient approximations by algorithmic methods such as SA and MP is of interest. As SA and MP are general methods that apply to general combinatorial optimization problems, it is expected that such methods are unlikely to outperform specially designed algorithms for specific problems~\cite{auger2011theory}. We thus view this work as being related more to the analysis of SA and MP and less to the study of the best approximation ratios that can be efficiently attained for the independent set problem. We note that proving lower bounds for MP and SA has proven difficult to achieve even for instances where the independent set problem (or equivalently the clique problem) is believed to be intractable such as sparse random graphs \cite{coja2015independent} and the planted clique problem~\cite{chen2022almost,jerrum1992large}. For both of these families of instances, it is not known how to prove lower bounds against MP and SA in their full generality. Tackling first the easier challenge of proving lower bounds for worst-case instances could be instrumental in proving more general lower bounds for MP with respect to sparse random graphs or the planted clique problem. 

\subsection{Our results}
As detailed in the previous section, our lower bounds apply to the problem of approximating the maximum independent set of a graph. Given a graph $G$, an independent set in $G$ is a subset of vertices that spans no edge. The cardinality of a maximum independent set in $G$ is denoted by $\alpha(G).$ Computing $\alpha(G)$ exactly or approximately is a classical NP-hard problem. 

The exact dynamics of the Metropolis process, Simulated Annealing, and the various variants we consider are introduced in \cref{sec:prelim}, and we refer the reader there for exact definitions. For now, let us mention that the main differences between the different algorithms lie in how the temperature, also called (inverse) fugacity, changes over the execution of the algorithm. 

In the classical Metropolis process, the temperature is fixed and does not change, while for Simulated Annealing there is some fixed schedule for decreasing the temperature over time. Some of our results also apply to a more general class of algorithms where the temperature can be chosen adaptively during the algorithm's execution. In the sequel, for simplicity we shall colloquially refer to all these algorithms as the Metropolis process (MP) and make sure to mention the different temperature schedules when relevant.

Our main results consist of exponential lower bounds on the time complexity of MP when approximating the value of $\alpha(G)$. We construct infinite families of graphs $G_1, G_2, \ldots, G_n, \ldots$ (with $G_n$ having $n$ vertices) such that the following holds. There is a function $p:\mathbb{N} \to (0,1)$ and a constant $\eta > 0$, such that if MP runs for fewer than $e^{n^{\eta}}$ iterations the probability it will find an independent set in $G_n$ larger than $p(n)\alpha(G_n)$ is at most $e^{-n^{\eta}}$. In other words, when run for less than an exponential number of steps, MP gives a multiplicative approximation of at most $p(n)$ to $\alpha(G_n)$.
As an instructive case, for general graphs, we show that one can take $p(n) = \frac{1}{n^{1-2\eps}}$, while $\alpha(G_n) = n^{1-\eps}$, for some $\eps > 0$ arbitrarily small. Thus, even though $G_n$ contains an independent set of nearly linear size, MP may struggle to even find an independent set of size $n^{\eps}$. All of our results hold when MP is initialized from the empty set. As previously noted, proving lower bounds for these algorithms when starting from a given state has proven challenging

\paragraph{Results for general graphs.} Our first main result is rather general and establishes lower bounds for the classical Metropolis process (with constant temperature) on graphs parametrized by their average degree.
In \cref{sec:technique} we outline the key ideas used in the proof and the complete proof can be found in \cref{sec:failure}.
\begin{theorem} \label{thm:sparse}
	Let $\{d_n\}_{n=1}^{\infty}$ satisfy $d_n \ge C$ for some large enough constant $C > 0$, and $d_n = o\left(\frac{n}{\log^2(n)}\right)$. There exists a sequence of graphs $\{G_n\}_{n\geq 0}$ satisfying:
	\begin{itemize}
		\item $G_n$ has $\Theta(n)$ vertices, the average degree of $G_n$ is  $\Theta(d_n)$, and $\alpha(G_n) = \Theta(\frac{n}{\log(d_n)})$.
		\item If $\{I_t\}_{t \geq 0}$ is the process of independent sets maintained by MP with any fixed\footnote{By fixed we mean that the temperature does not change during the algorithm. The temperature parameter may depend on $n$.} temperature, 
		$$\PP\left(\max\limits_{t \leq \exp\left(\frac{n}{C' d_n}\right)}|I_t| \geq C'\frac{\log(d_n)}{d_n}\alpha(G_n) \right)\leq \exp\left(-\frac{n}{d_n\log(d_n)}\right),$$
		where $C' > 0$ is a universal constant.
	\end{itemize}
\end{theorem}

Let us unpack \cref{thm:sparse} and consider the extreme cases for the average degree. The largest degree we can take and still obtain super-polynomial bounds is $d_n = \frac{n}{\mathrm{polylog}(n)}$. In this case, \cref{thm:sparse} guarantees that $G_n$ has an independent set of nearly linear size $\frac{n}{\mathrm{polylog}(n)}$. However, with any fixed temperature, if MP runs for only a polynomial number of steps, it will fail to find an independent set of even polynomial size and will only result in a set of size $\mathrm{polylog}(n)$. 
To get exponential lower bounds we can slightly lower the average degree and take $d_n = n^{1-\eps}$ for any fixed $\eps > 0$. For these slightly sparser graphs, if MP runs for $\exp(n^\eps)$ iterations, it will only find a set of size at most $\tilde O(n^{\eps})$. To put this result in context, as was mentioned above, it is known that it is NP-hard to approximate the maximum independent set to within an $O\left(\frac{1}{n^{1-\eps}}\right)$ factor~\cite{hastad1996clique,zuckerman2006linear,khot2006better} for any $\eps \in (0,1)$. Thus, the exponential lower bound is predicted by this hardness result and should be seen as an unconditional proof of this prediction for MP.

\cref{thm:sparse} also applies when the average degree of the graph is a constant, that does not depend on the number of vertices. For these sparse graphs, there is extensive literature surrounding the question of approximating the maximum independent set~\cite{halldorsson1994improved,bhangale2022ug,halperin2002improved,bansal2015lovasz}. Hence, it is interesting to study the approximation achieved by MP (with polynomial running time) for sparse graphs. 
\cref{thm:sparse} allows to take $d_n \equiv d$, for some large enough constant average degree $d > 0$ and obtain a sparse graph. For our sparse graphs, MP will only find an $O\left(\frac{\log(d)}{d}\right)$ approximation of $\alpha(G)$. As an algorithmic counterpart to our lower bound, the randomized greedy algorithm will find an independent set of expected size at least $\frac{n}{d+1}$. Below, in \cref{sec:prelim} we explain how the randomized greedy algorithm can be instantiated as an MP algorithm, which shows that our lower bound is tight up to the $\log(d)$ factor. 
In \cref{sec:related} we discuss some more related results pertaining to the performance of Metropolis process in sparse Erd\"os-R\'enyi graphs.
\paragraph{Simulated Annealing in dense graphs.}
To go beyond the classical MP, and allow the temperature to change over time, we specialize \cref{thm:sparse} to denser graphs. A key appealing feature of our result in this case is that the theorem applies to any sequence of temperatures. In particular, the sequence can be adaptive (see \cref{sec:prelim} for the exact meaning of an adaptive sequence) and may be changed adversarially during the algorithm's execution.
\begin{theorem} \label{thm:dense}
	For every $\eps \in (0,\frac{1}{2})$, there exists a sequence of graphs $\{G_n\}_{n\geq 0}$ satisfying:
	\begin{itemize}
		\item $G_n$ has $\Theta(n)$ vertices and $\alpha(G_n) = \Theta(n^{1-2\eps})$.
		\item For any temperature schedule, which can be adaptive, if $\{I_t\}_{t \geq 0}$ is  the process of independent sets maintained by MP, then
		$$\PP\left(\max\limits_{t \leq e^{n^\eta}}|I_t| \geq 2n^{\eps}\right) \leq e^{-n^\eta},$$
		for some $\eta >0$.
	\end{itemize}
\end{theorem}
As mentioned above, when the temperature is some predetermined sequence that decreases over time, the MP algorithm is also known as Simulated Annealing. Therefore, by considering this temperature scheduling, \cref{thm:dense} bounds the best approximation ratio SA can achieve. As discussed, this bound precisely matches the best-known results that follow from NP-hardness and again serves as proof of their prediction. The theorem goes beyond SA and, unsurprisingly, shows that there is no way to change the temperature schedule (even if one is allowed to make changes during execution) to go beyond the hurdle suggested by NP-hardness results. Adaptive changes to the temperature in the SA algorithm have been suggested before~\cite{ingber1989very}. We are not aware of previous rigorous results about the benefits or limitations of adaptivity when using these methods to efficiently solve NP-hard optimization problems. 

The proof of \cref{thm:dense} appears in \cref{sec:failure}.

\paragraph{Results for bipartite graphs.}
The hard instances for \cref{thm:sparse,thm:dense} will be introduced in \cref{sec:failure}. A key feature of our construction is that the hard instances are constructed from bipartite graphs. These graphs are then augmented by blowing up some vertices into cliques, losing the bipartite structure. Given our construction, it is also interesting to study the performance of MP on general bipartite graphs. The point is that, in this case, there exists a simple linear time algorithm to obtain a $\frac{1}{2}$ approximation of $\alpha(G)$ by finding a bipartition. Furthermore, the standard linear programming relaxation for the independent set problem can recover the exact size of $\alpha(G)$ in polynomial time. Keeping in mind the tractability of the problem for bipartite graphs, it seems natural to expect that there exists some variant of MP that will fare similarly in these instances. On the contrary, our next result shows that in general, MP with any temperature schedule fails to come close to the performance of the mentioned algorithms.
\begin{theorem} \label{thm:bipartite}
    Let $d_n \leq \frac{\log(n)}{100}$. There exists a sequence of \emph{bipartite} graphs $\{G_n\}_{n \geq 0}$ satisfying.
    \begin{itemize}
        \item $G_n$ has $\Theta(n)$ vertices, average degree $\Theta(d_n)$.
        \item For any temperature schedule, which can be adaptive, if $\{I_t\}_{t \geq 0}$ is the process of independent sets maintained by MP, then
		$$\PP\left(\max\limits_{t \leq e^{n^\eta}}|I_t| \geq (4 + o(1))\frac{\log(d_n)}{d_n}\alpha(G_n)\right) \leq e^{-n^\eta},$$
		for some $\eta >0$.
    \end{itemize}
\end{theorem}
\cref{thm:bipartite} implies that there exist $n$-vertex bipartite graphs for which SA cannot efficiently approximate the size of the largest independent set within a ratio better than $O\left(\frac{1}{\log(n)}\right)$. We did not attempt to optimize the hardness ratio as our main point is that MP, with any temperature (thus also covering SA), fails to find an approximate solution even in instances where the independent set problem is tractable. It is possible that stronger inapproximability results hold for SA: It might be that it fails to efficiently approximate the independent set problem in $n$-vertex bipartite graphs within a ratio larger than $1/n^c$ for some $c \in (0,1)$. Studying this question is left for future work.

The proof of \cref{thm:bipartite} appears in \cref{sec:bipartite}. 

\paragraph{Performance of Simulated Annealing on trees.}
\cref{thm:bipartite} shows that, even on tractable instances, MP and its variants achieve significantly worse approximation when compared to polynomial-time algorithms. In our final hardness result, we further emphasize this point by considering the, arguably, easiest class of graphs for the independent set problem: trees. Trees are a strict and simpler sub-class of bipartite graphs. A simple greedy algorithm will return the maximum independent set in polynomial time~\cite{kleinberg2006algorithm}. In our next theorem, we give a complete characterization of the performance of MP, with a non-adaptive temperature schedule (like in SA), on trees. In particular, we show that MP is not competitive with polynomial-time algorithms:
with less than an exponential number of iterations, there are trees where it will fail to find \emph{the} maximum independent set. We complement this hardness result by establishing that MP can return an arbitrarily good \emph{approximation} to $\alpha(G)$ in polynomial time. 
\begin{theorem} \label{thm:forest}
    The following hold:
       \begin{itemize}
          \item There exists a sequence of $n$-vertex trees $\{F_n\}_{n \ge 0}$ such that for any constant $\eta \in (0,\frac{1}{4})$, if $\{I_t\}_{t \geq 0}$ is the process of independent sets maintained by MP with any non-adaptive sequence of fugacities,
            $$\PP\left(\max\limits_{t \leq \exp(n^\eta)}|I_t| = \alpha(F_n)\right) \leq e^{-\Omega(\sqrt{n})}.$$
        \item For any constant $\eps \in (0,1)$ and any $n$-vertex forest $F_n$, there exists $\lambda = \lambda(\eps,n)$ such that, if $\{I_t\}_{t \geq 0}$ is the process of independent sets maintained by MP with fixed fugacity $\lambda$,
             $$\PP\left(\max\limits_{t \leq \mathrm{poly}(n)}|I_t| \geq (1-\eps)\alpha(F_n)\right) \geq 1 - o(1).$$
        \end{itemize}
\end{theorem}
The proof of \cref{thm:forest} can be found in \cref{sec:tree}.

\paragraph{Greedy algorithms vs.~MP.}
As a final remark, one may wonder if there are instances where MP (when restricted to polynomially many iterations) is superior to the well-studied greedy~\cite{halldorsson1994greed} and randomized greedy~\cite{gamarnik2010randomized} algorithms for the maximum independent set problem. If this was not the case, one could prove lower bounds for MP by constructing hard instances for greedy algorithms. While greedy algorithms were shown to achieve comparable results to MP for certain problems~\cite{jerrum1998metropolis,carson2001hill} we provide simple examples in \cref{sec:greedyvsmetro} where greedy algorithms achieve significantly worse approximation compared to MP in approximating the independent set problem.

\subsection{Proof approach} \label{sec:technique}

We first explain our proof approach for \cref{thm:sparse,thm:dense}  which establish the failure of the Metropolis process for finding large independent sets in graphs of given edge density characterized by the average degree. 
We prove this by carefully constructing a family of random graphs. 
Naively, we would hope to use Erd\"os-R\'enyi random bipartite graphs which are significantly unbalanced as bad instances. More specifically, suppose that the vertex set is partitioned into $V = L\cup R$ and every edge connects one vertex from $L$ and one from $R$. Ideally, we want to have $|L| \ll |R|$, and be able to show that the Metropolis process is more likely to pick up vertices in $L$, and thus reaches independent sets mostly contained in $L$. 
However, a moment of thought immediately shows this cannot be the case. Especially, if in each step the Metropolis process picks a vertex \emph{uniformly at random}, then vertices in $R$ are more likely to be chosen and MP will get independent sets of large overlap with $R$, which is nearly optimal. 

Instead of simply using an Erd\"os-R\'enyi random bipartite graph, we further augment it with a blowup construction. More specifically, we replace each vertex $u \in L$ from the smaller side with a clique of size $\ell$ and connect this clique to all neighbors of $u$. We pick $\ell$ sufficiently large so that $\ell |L| \gg |R|$. 
This immediately provides two advantages. First, in every step, MP is more likely to choose vertices from the new $L$,  now a disjoint union of $|L|$ cliques of size $\ell$, because $\ell|L| \gg |R|$. Second, independent sets of the blowup graph are in one-to-one correspondence to independent sets of the original graph since each clique can have at most one occupied vertex. Furthermore, it is much more difficult to remove an occupied vertex in order to make the corresponding clique unoccupied, since the MP has to pick the correct occupied vertex among all vertices in the clique. These properties are formalized in \cref{lem:generalMP}.

We can then argue that the MP will not be able to find a large independent set for these blowup graphs with polynomially many steps. Suppose that $|L|=n$ and $|R| = kn$ and that $\ell \gg k \gg 1$. 
Then, after a suitable burn-in phase, we will show that the MP will pick at most a tiny fraction of vertices in $R$, and at least a constant fraction of vertices in $L$. Thus, within the burn-in phase, MP only reaches independent sets mostly contained in $L$. In particular, there is a set $L_1\subseteq L$ of at least $n/10$ occupied vertices in $L$, and a set $R_0 \subseteq R$ of at least $(k-1)n$ unoccupied vertices in $R$, at the end of the burn-in phase (see \cref{lem:burn-in}). These vertices induce a smaller Erd\"os-R\'enyi random bipartite graph and the MP on it with the initialization $L_1$ will contain vertices mostly from $L_1$ and barely from $R_0$, via simple conductance (i.e., bottleneck) arguments. 
Thus, within polynomially many steps the MP cannot reach independent sets with too many vertices from $R_0$. 
In fact, the obtained independent set contains at most $n$ vertices from $R_0$ with high probability (see \cref{lem:after-burn-in}), and consequently has size at most $3n$ since $|L|=n$ and $|R \setminus R_0| \le n$. 
Meanwhile, $R$ is an independent set of size $kn \gg 3n$, exhibiting the failure of MP.
We remark that the whole argument works even for constant $k,\ell$ and edge density $O(1/n)$ so that the average degree is constant, though all ``$\gg$'' will be replaced with explicit inequalities.

Our construction of \emph{bipartite}, appearing in \cref{thm:bipartite}, is based on the $t$-blowup operation. In this blowup, every vertex is replaced by an independent set (``cloud") of size $t$ and two clouds that correspond to neighboring vertices (before the blowup) are connected by a complete bipartite graph. The main observation is that once a vertex from a cloud is chosen, MP is much more likely to keep adding vertices to the cloud as opposed to deleting vertices from it. By taking many (identical) duplicates of the $t$-blowup of an initial bipartite graph, we get that for a large fraction of the duplicates no cloud is deleted (assuming a vertex from it is chosen by the algorithm) in polynomial time, hence resulting in essentially the randomized greedy algorithm where deletions do not occur. To conclude the proof we need to provide a \emph{bipartite} graph for which randomized greedy does badly: This can serve as the ``base graph" on which we perform the blow-up and duplication. We prove that the random balanced bipartite graph of size $2n$ with edge probability $d/n$, for a large enough constant\footnote{Our lower bounds can be extended to $d=O(\log(n))$.} $d > 0$, is with high probability a hard instance for randomized greedy. Our proof follows a martingale argument and may be of independent interest. The limitation of greedy algorithms for coloring (and implicitly independent set) has been observed before~\cite{kuvcera1991greedy} for random multipartite graphs $m$-vertex graphs where the partition includes $m^{\Omega(1)}$ parts. We are not aware of a previous hardness result for the randomized greedy for approximating the size of the maximum independent set in bipartite graphs.  

For trees, the core ingredient of the hard instance is a ``star-shaped'' tree composed of a root $r$ connected to $k$ nodes 
$a_1, \ldots, a_k$, and each node $a_i$ has a single leaf neighbor $b_i$. The unique maximum independent set consists of $r$ together with all the $v_i$'s. In the first $m=k^{1/2-\eps}$ iterations, MP will add roughly $m/2$ neighbors of $r$. Let $I$ denote the set of indices $i$ for these chosen neighbors. The crux of the argument is to track the configuration of the branches $I$ and show that they behave roughly like an i.i.d.\ collection of random variables supported on three states (free, $a_i$ chosen, $b_i$ chosen) where the probability $a_i$ is chosen is at least $1/4$. It follows that with high probability, \emph{some} $a_i$ will be occupied for exponential time, blocking the root $r$ from ever being added. The argument is completed by duplicating many copies of the hard tree, ensuring the probability that the process runs for less than exponentially many iterations is exponentially small. Finally, the forest can be made into a tree by connecting all the roots of the trees to a single additional vertex. The upper bound, showing that MP can efficiently approximate the optimum solution in a tree within a factor of $1-\eps$ for arbitrary $\eps \in (0,1)$, is a simple consequence of the rapid mixing of MP for trees~\cite{chen2023combinatorial,eppstein2021rapid}, taking $\lambda$ to be a sufficiently large constant to ensure the partition function is concentrated on large independent sets. 
\section{Further Related work} \label{sec:related}

One of the earliest works studying lower bounds for the Metropolis process is due to Jerrum~\cite{jerrum1992large}. In his work, he considered the planted clique problem where one seeks to find a clique of size $n^{\beta}$ for $\beta \in (0,1)$ planted in the Erd\"os-R\'enyi random graph $G(n,\frac{1}{2})$. Jerrum proved using a conductance argument the existence of an initial state for which the Metropolis process for cliques fails to find a clique of size $(1+\varepsilon)\log n$ assuming $\beta<\frac{1}{2}$ so long as it is executed for less than $n^{\Omega(\log (n))}$ iterations. 

Several open questions were raised in~\cite{jerrum1992large} regarding whether one could prove the same lower bound for MP when initialized from the empty set, whether the same lower bound holds for arbitrary $\beta<1$ (as opposed to $\beta<\frac{1}{2}$) and whether similar lower bounds could be extended to SA as opposed to MP. The recent paper~\cite{chen2022almost} made the first substantial progress towards answering Jerrum's question; when the inverse temperature satisfies $\lambda \leq O(\log (n))$, the super-polynomial lower bounds hold with respect to MP initialized from the empty set even when $\beta<1$. Under the same assumptions, the result of~\cite{chen2022almost} also applies to SA (initialized from the empty set) for a certain temperature scheduling termed simulated tempering~\cite{marinari1992simulated}. The lower bounds in~\cite{chen2022almost} do not rule out that MP could solve the planted clique problem (even for $\beta <1/2$) when the inverse-temperature is set to $C \log n$ for a suitable constant $C$. Establishing a lower bound on the running time of MP for every possible temperature is mentioned as an open question in~\cite{chen2022almost}.

It is natural to compare the results in~\cite{chen2022almost} to our lower bounds, such as \cref{thm:sparse,thm:dense} which apply without any restriction on the temperature and establish exponential (as opposed to quasi-polynomial) time lower bounds. Moreover in the dense case, which is the analog of $G(n,\frac{1}{2})$, our bounds go beyond the restriction of simulated tempering and cover any sequence of temperatures. It should be noted however that~\cite{chen2022almost} focused on resolving Jerrum's questions, while our work is geared towards proving general lower bounds. Thus, for example, in the planted clique problem a quasi-polynomial lower bound is the best one can hope for, as MP can be shown to solve the problem with high probability in $n^{O(\log (n))}$ iterations.  In contrast, to prove our lower bounds we choose carefully crafted instances of random distributions on graphs, which allows for more flexibility in establishing lower bounds.

In~\cite{coja2015independent} exponential lower bounds were proven with respect to the \emph{mixing time} of MP for independent sets in sparse random graphs $G(n,\frac{d}{n})$ assuming $d$ is a large enough constant.
Observe however that lower bounds on the mixing time do not imply lower bounds for the time it takes MP to encounter a good approximation of $\alpha(G)$. It is still open whether MP (with empty set initialization) fails to find an independent set of size $(1+\varepsilon) \frac{\log(d)}{d}n$ in polynomial time in $G(n,\frac{d}{n})$. While the setting and proofs in~\cite{coja2015independent} are different from those in our paper, a common theme in both papers is that a barrier for MP is the need to delete many vertices from a locally optimal solution in order to reach a superior approximation to the optimum.

Few additional lower bounds are known for the time complexity of SA when used to approximately solve combinatorial optimization problems. Sasaki and Hajek~\cite{sasaki1988time} proved that, for certain instances, when searching for a maximum matching in a graph, the running time of SA can be exponential in the number of vertices (even when initialized from the empty set). Let us stress the fact that the lower bound in \cite{sasaki1988time} concerns finding an \emph{exact} solution, rather than approximating the solution. Moreover, their family of hard instances cannot be used to prove SA requires super-polynomial time to approximate maximum matching (and hence the independent set problem) within a factor of $\alpha$ for a fixed $\alpha \in (0,1)$, as they prove that for every fixed $\varepsilon \in (0,1)$ MP yields a $1-\varepsilon$ multiplicative approximation for the maximum matching problem in polynomial time. A different proof showing that MP can find a $1-\varepsilon$ approximation for maximum matching in polynomial time, was discovered later in ~\cite{jerrum1989approximating}. 

In~\cite{wegener2005simulated} it is shown, by providing a family of hard instances, that SA cannot find in polynomial time a $1+o(1)$ multiplicative approximation for the minimum spanning tree problem. However, this result cannot be extended in a substantial way to show that SA fails to find a $1+\eps$ approximation for \emph{fixed} $\eps>1$: For the MST problem (with non-negative weights), it was proven in~\cite{doerr2022simulated} that SA can find a $1+\varepsilon$ approximation for the optimal solution in polynomial time for any fixed $\varepsilon >0 $, extending an earlier result of~\cite{wegener2005simulated}.
In another direction,~\cite{bhatnagar2004torpid} proves exponential lower bounds on the mixing time of SA-based algorithms, which can be seen as a crucial hindrance for finding approximate solutions. 

In terms of approximation ratios that can be achieved in polynomial time by SA, \cite{johnson1989optimization,johnson1991optimization} provide extensive empirical simulations for SA when applied to NP-hard optimization problems such as min bisection and graph coloring. In terms of rigorous results regarding the approximation ratio achieved efficiently by SA in the worst case,~\cite{gharan2011submodular} provides an algorithm inspired\footnote{The authors in~\cite{gharan2011submodular} mention that ``We should remark that our algorithm is somewhat different from a direct interpretation of simulated
annealing." For more details, see~\cite{gharan2011submodular}.} by SA that achieves in polynomial time a $0.41$-approximation for unconstrained submodular maximization and a $0.325$-approximation for submodular maximization subject to a matroid independence constraint. In~\cite{sorkin1991efficient} it is proven that certain properties associated with the energy landscape of a solution space may lead SA to efficiently find the optimal solution. 

Compared to SA, somewhat more is known regarding lower bounds for MP. Sasaki~\cite{sasaki1991effect} introduced families of $n$-vertex instances of min-bisection and the traveling salesman problem (TSP) where MP requires exponential time to find the optimal minimum solution. Sasaki's proof is based on a ``density of states" argument. The main step is to show, via a conductance argument, that when the number of optimal solutions is smaller by an exponential multiplicative factor than the number of near-optimal solutions, there exists an initialization where the expected hitting time of the optimal solution is exponential in $n$. It is explicitly mentioned in~\cite{sasaki1991effect} that the proof methods do not imply lower bounds for SA. While our hard instance for trees has the property that the number of optimal solutions of value OPT is smaller by an exponential factor than the number of solutions of value OPT$-1$, our proof that SA requires exponential time to find the optimal solution differs from the proofs in~\cite{sasaki1991effect} (indeed our proof applies for SA whereas the proof in~\cite{sasaki1991effect} applies only for the fixed temperature case).  
The paper~\cite{meer2007simulated} contains constructions of instances of the traveling salesman problem (TSP) where MP takes exponential time to find the optimal solution but SA with an appropriate cooling schedule finds a tour of minimum cost in polynomial time. Both~\cite{sasaki1991effect,meer2007simulated} prove lower bounds for \emph{exact} computation of the optimum and do not prove lower bounds for algorithms approximating the optimum. 
An informative survey of these and additional results related to SA and MP can be found in~\cite{auger2011theory}.

There is an extensive literature on efficient approximation algorithms for the independent set problem. As mentioned, for general $n$-vertex graphs, it is NP-hard to approximate $\alpha(G)$ within a factor of $\frac{1}{n^{1-\varepsilon}}$ for any $\varepsilon \in (0,1)$~\cite{hastad1996clique,zuckerman2006linear}. Under a certain complexity-theoretic assumption it was shown in~\cite{khot2006better} 
that approximating the size of an independent set in $n$-vertex graphs within a factor larger than $2^{(\log n)^{3/4+\gamma}}/n$ (for arbitrary $\gamma>0$) is impossible. The current best efficient approximation algorithm for the independent set problem achieves a ratio of $\Omega(\frac{\log(n)^3}{n \log \log (n)^2})$~\cite{feige2004approximating}. For graphs with average degree $d$ it has long been known that a simple greedy algorithm achieves an approximation of $\Omega\left(\frac{1}{d}\right)$. This bound has been gradually improved~\cite{halldorsson1994improved,halperin2002improved}, and the state of the art~\cite{bansal2015lovasz} is an algorithm based on the Sherali-Adams hierarchy achieving an approximation ratio of $\Tilde{\Omega}\left(\frac{\log(d)^2}{d}\right)$
for graphs of maximum degree $d$ (the $\Tilde{\Omega}$ sign hides $\mathrm{poly}(\log \log d)$ multiplicative factors). The running time of this algorithm is polynomial in $n$ and exponential in $d$. This matches, up to $\mathrm{poly}(\log \log d)$ factors, the lower bound in \cite{bhangale2022ug} showing that obtaining an approximation ratio larger than $O\left(\frac{\log ^2(d)}{d}\right)$ is NP-hard in graphs of maximum degree $d$ (assuming $d$ is a constant independent of $n$). In contrast to these lower bounds our lower bounds for the time complexity needed to find approximate solutions apply to specific algorithms (MP and SA). On the other hand, our results are unconditional and also apply to instances (such as bipartite graphs) where polynomial-time algorithms are known to find the optimal solution. 

\section{The Metropolis Process} \label{sec:prelim}

In this section, we introduce the different variants of the Metropolis process for which we prove lower bounds. Our description differs slightly from the algorithm described in the introduction; we will work on a logarithmic scale for the temperature and parametrize the algorithm by the inverse of the temperature, also called \emph{fugacity}. Ultimately this is a choice for convenience and the two descriptions are equivalent.

All considered algorithms for finding (random) large independent sets of a given input graph fall under one very general stochastic process, called the \emph{Universal Metropolis Process (UMP)}. Among others, UMP incorporates the Randomized Greedy algorithm, the Metropolis process, and the Simulated Annealing process. 
Let $G=(V,E)$ be an input graph, and $T \in \N$ be the number of steps of the process.
For each $t \ge 1$, let $\lambda_t\in [1,\infty]$ be  the fugacity, or inverse-temperature, at time $t$. 
We refer to the collection $\{\lambda_t\}_{t=1}^\infty$ of all fugacities as the \emph{fugacity schedule} of the process. With these definitions, the Universal Metropolis Process is described by \cref{alg:MP,alg:MP-1step}.

\begin{algorithm} [t]
\caption{Universal Metropolis Process (UMP)}\label{alg:MP}
\begin{algorithmic}[1]
\Input $G=(V,E)$ a graph, $T$ the number of steps, 
$\{\lambda_t\}_{t=1}^T$ a fugacity schedule

\State $I_0 \gets \emptyset$;
\Comment{Empty IS initialization}

\For{$t \gets 1 \text{~to~} T$}
\State Sample a vertex $v_t \in V$ u.a.r.;
\State Sample a real $\zeta_t \in [0,1]$ u.a.r.; 
\State $I_t \gets \textsc{Update}(I_{t-1},v_t,\zeta_t, \lambda_t)$;
\Comment{\cref{alg:MP-1step}}
\EndFor

\Output $I_{t^*}$ where $t^* = \argmax\limits_{t \in \{1,\dots,T\}} |I_t|$
\end{algorithmic}
\end{algorithm}

\begin{algorithm} [t]
\caption{$\textsc{Update}(I,v,\zeta,\lambda)$}\label{alg:MP-1step}
\begin{algorithmic}[1]

\Input $I$ an IS, $v$ a vertex, $\zeta \in [0,1]$, $\lambda \in [1,\infty]$

\If{$v \notin I$}
    \State $
    I' \gets
    \begin{cases}
    I \cup v, & \text{if $I \cup v$ is an IS}; \\
    I, & \text{otherwise}.
    \end{cases}
    $
\Else ~(i.e., $v \in I$)
	\State $
	I' \gets 
	\begin{cases}
	I \setminus v, & \text{if $\zeta \le \frac{1}{\lambda}$}; \\
	I, & \text{otherwise}.
	\end{cases}
	$
\State 
\EndIf
\Output $I'$
\end{algorithmic}
\end{algorithm}

We observe that:
\begin{itemize}
    \item If $\lambda_t \equiv \infty$ for all $t>0$, then UMP corresponds to the Randomized Greedy algorithm.
    
    \item If $\lambda_t \equiv \lambda$ for all $t>0$ and for some $\lambda \in [1,\infty)$, independent of $t$, then UMP corresponds to the Metropolis process with fugacity $\lambda$, which is a Markov chain whose stationary distribution is the Gibbs distribution $\mu$ of the hardcore model on $G$ (i.e., $\mu(I) \propto \lambda^{|I|}$ for each independent set $I$ of $G$).
    
    \item If $\{\lambda_t\}$ forms a predetermined non-decreasing sequence, i.e. $\lambda_1 \le \lambda_2 \le \cdots$, then UMP corresponds to the Simulated Annealing algorithm with fugacity schedule $\{\lambda_t\}$.

\end{itemize}

Of course, UMP goes beyond these three well-known cases and allows for, say, non-monotone fugacities schedules. In general, if $\{\lambda_t\}$ is some arbitrary deterministic sequence or if it is random, but independent from the randomness of $\{I_t\}$, we shall call it a \emph{non-adaptive} schedule. On the other hand, if $\lambda_t$ can depend on $\{I_s\}_{s=1}^{t-1}$ (and so it is necessarily random), we call the schedule \emph{adaptive}.

\section{Failure of Metropolis Process for Finding Large Independent Sets} \label{sec:failure}

In this section, we prove \cref{thm:sparse,thm:dense}. Before delving into the proof, we first introduce a family of random graphs that shall serve as hard instances for the theorems. We then show that the process on independent sets have similar behavior during the early stages, which we term the burn-in phase. The proofs will diverge when considering sparse and dense graphs.

\subsection{Construction of the hard instance}
\label{subsec:construction}
We first define the family of random graphs.
\paragraph{Parameters.}
Our family of random graphs is parameterized by three positive integers $n,k,\ell \in \N^+$, and a real $p \in (0,0.1)$ where $p$, $\ell$, and $k$ are all functions of $n$, and we enforce the following relations: assume $n,k$ are sufficiently large,
\begin{equation} \label{eq:relations}
 \frac{50 \log k}{n} \leq p = o(1)
 \qquad\text{and}\qquad
 \ell \ge 10 kpn.
\end{equation}
As will be soon apparent our hard instances are constructed from random bipartite graphs. With this in mind, $n$ should be viewed as, roughly, the number of vertices, $k$ and $\ell$ then encode the imbalance between the two sides, and $p$ is the edge density.

\paragraph{Base random graph $\mathcal{B}(n, kn, p)$.}
Define the base graph $B=(V_B,E_B)$ to be a random bipartite graph $\mathcal{B}(n, kn, p)$, namely:
\begin{itemize}
\item $V_B = L\cup R$ where $|L| = n$ and $|R| = kn$;
\item For every $u \in L$ and $v \in R$, include the edge $\{u,v\}$ in $E_B$ with probability $p$ independently .
\end{itemize}

\paragraph{The random instance $\mathcal{G}(n,\ell,k,p)$.}
Define the hard instance $G=(V_G,E_G)$ as follows: 
\begin{itemize}
\item Sample a base graph $B=(V_B=L\cup R,E_B) \sim \mathcal{B}(n, kn, p)$;
\item Replace every vertex $u \in L$ with a clique $K_u$ of size $\ell$ 
(namely, $V_G$ consists of one independent set $R$ of size $kn$ and $n$ cliques each of size $\ell$);
\item For every edge $\{u,v\} \in E_B$ where $u \in L$ and $v \in R$, include the edges $\{u',v\}$ for all $u' \in K_u$ (namely, fully connect the clique $K_u$ with $v$).
\end{itemize}
The following facts, which characterize the basic properties of the random graph $\mathcal{G}(n,k)$, can be readily seen from the construction.
\begin{fact} \label{fct:fact}
Let $G=(V_G,E_G)$ be a random graph generated from $\mathcal{G}(n,\ell,k,p)$.
\begin{itemize}
\item The size of $G$ is $N := |V_G| = kn + \ell n = \Theta(\ell n)$.
\item With high probability, the average degree of $G$ satisfies $D := \frac{2|E_G|}{|V_G|} = \Theta(\ell)$. 
\item The independence number of $G$ is $\alpha(G)= \Omega(kn)$.
\end{itemize}
\end{fact}
\begin{proof}
The number of vertices is immediate.
For the average degree, we first note that the expected number of edges is given by
\begin{align*}
\E[|E_G|]
= p \cdot kn \cdot \ell n + \frac{1}{2} \ell(\ell-1) \cdot n
= \Theta(\ell^2 n)
\end{align*}
where we use $\ell \ge 10 kpn$.
Thus, the expected average degree is given by
\begin{align*}
\E[D] = \frac{2\E[|E_G|]}{|V_G|} 
=\Theta(\ell),
\end{align*}
and the average degree concentrates around it with high probability.
Finally, $\alpha(G) \ge kn$ since $R$ is an independent set of size $kn$.
\end{proof}

\begin{remark}\label{rmk:max-degree}
If $k,\ell = O(1)$ and $p = O(1/n)$, then the size of $G$ is $N=\Theta(n)$ and the average degree of $G$ is $\Theta=O(1)$ with high probability. Moreover, the maximum degree of $G$ is $\Theta(\frac{\log n}{\log\log n})$ with high probability, which follows from similar arguments as in e.g.\ \cite[Theorem 3.4]{frieze2016introduction}.
\end{remark}
 
The main idea that we shall utilize is that the number of vertices in $\cup_{u \in L} K_u$ is larger by a pre-determined factor ${\frac{\ell}{k}}$ than the size of the (nearly largest) independent set $R$. Thus one should expect UMP to reach independent sets mostly contained in $\cup_{u \in L} K_u$ which have size at most $n$. As we shall show these independent sets serve as a natural bottleneck for the UMP to reach $R$. 

\subsubsection{Continuous-time Universal Metropolis Process} 
For our analysis, it will be slightly more convenient to work with a continuous-time version of UMP defined on the base graph $\mathcal{B}(n, kn, p)$. 
As we will see in \cref{lem:generalMP} soon, these two versions are ultimately equivalent to each other. 
To introduce the continuous-time process, let $G=(V,E)$ be a graph, and define the following parameters:
\begin{itemize}
\item $\rho: V \to \R^+$: vertex update rate, for determining the rate of updates at each vertex;
\item $\lambda_t: V \to [1,\infty)$: vertex fugacity function at time $t \in \R_{\ge 0}$.
\end{itemize}
Recall that the discrete-time Universal Metropolis Process is given in \cref{alg:MP}.
The continuous-time version is described as follows.  
\begin{itemize}
    \item At each vertex $v \in V$, we place an independent Poisson clock (i.e., Poisson point process) with rate $\rho(v)$;
    \item Whenever a Poisson clock of some vertex $v_t$ rings at time $t \in \R_{\ge 0}$, we sample a real $\zeta_t \in [0,1]$ u.a.r.\ and update the current IS by $I_t \gets \textsc{Update}(I_t,v_t,\zeta_t,\lambda_t(v_t))$ using \cref{alg:MP-1step}.
\end{itemize}
We remark that we use the continuous-time UMP only as an auxiliary process for the purpose of analysis, and thus we do not consider implementation of it. 
If the rates $\rho(v)$ are uniform for all $v$, then this corresponds to the scenario that all vertices are equally likely to be updated, which is the standard case. 
If for each $t$ the vertex fugacity function $\lambda_t(v)$ is uniform for all $v$, then this corresponds to the standard setting for UMP. 
However, for our analysis of UMP we require a non-uniform update distribution on vertices and a non-uniform fugacity function of vertices.

We now make the formal connection between the discrete-time Universal Metropolis process on $\mathcal{G}(n,\ell,k,p)$ and the continuous-time process on $\mathcal{B}(n, kn, p)$.

\begin{lemma} \label{lem:generalMP}
For any pair of graphs $(B,G)$ where $B$ is a bipartite graph and $G$ is constructed from $B$ as described in \cref{subsec:construction}, 
and for any real $T \in \R_{\ge 0}$, 
the discrete-time UMP on $G$ and the continuous-time UMP on $B$ are equivalent to each other under the following setups:
\begin{itemize}
\item Run the discrete-time UMP on $G$ with fugacity schedule $\{\lambda'_{t'}\}_{t'=1}^\infty$ for $T'$ steps where $T' \in \N$ is an independent Poisson random variable with mean $(k+\ell)nT$. Let $I'$ denote the (random) output independent set at time $T'$. 
\item Run the continuous-time UMP on $B$ up to time $T$ with 
\[
\rho(v) = 
\begin{cases}
\ell, & v \in L\\
1, & v \in R
\end{cases}
\]
and for all $t \in \R_{\ge 0}$, 
\[
\lambda_t(v) = 
\begin{cases}
\ell\lambda_{t'}(v), & v \in L\\
\lambda_{t'}(v), & v \in R
\end{cases}
\]
where $t' \in \N$ is the total number of clock rings at all vertices in $[0,t]$.
Let $I$ denote the (random) output independent set at time $T$. 
\end{itemize}
Let $\varphi$ be a mapping which maps an independent set $I'$ of $G$ to an independent set $\varphi(I')$ of $B$ by replacing a vertex from $K_u$ in $I'$ with $u$ in $\varphi(I')$ for each $u \in L$. 
Then $\varphi(I')$ and $I$ have the same distribution.
\end{lemma}

\begin{proof}
It is straightforward to verify that the discrete-time UMP on $G$ is equivalent to the following generalized version of discrete-time UMP on $B$:
In each step $t$, we sample a vertex $v_t \in V_B$ from a distribution $p$ defined as $p(v) = \frac{\ell}{(k+\ell)n}$ for $v \in L$ and $p(v) = \frac{1}{(k+\ell)n}$ for $v \in R$ (instead of u.a.r.) and sample a real $\zeta_t \in [0,1]$ u.a.r.; 
Then we update $v_t$ using a non-uniform fugacity defined as $\lambda_t(v) = \ell \lambda_t$ for $v \in L$ and $\lambda_t(v) = \lambda_t$ for $v \in R$. 
The equivalence follows from that we replace each vertex $u \in L$ in $B$ with a clique $K_u$ of size $\ell$ in $G$ and fully connect $K_u$ with the neighborhood of $u$.
Then, the lemma follows from that such discrete-time UMP on $B$ is equivalent to the continuous-time UMP on $B$ as described in the lemma, which can be readily verified.
\end{proof}

In the rest of this section, we consider only the continuous-time UMP on $B$ and simply refer to it as MP.

\subsubsection{A simple lemma for random processes on random graphs}
The last component we shall need, before analyzing the dynamics of MP on our hard instances, is a simple technical lemma, which implies the existence of a hard instance from the hardness of a random distribution.
Let $\mathcal{Q} = \mathcal{Q}_T$ denote the randomness of MP on $B$ up to a fixed time $T$, which consists of the records of all rings of Poisson clocks and all uniformly random numbers in $[0,1]$ that are used in updates. 
Importantly, $\mathcal{Q}$ is independent of (the random edge set of) the random bipartite graph $B$ generated from $\mathcal{B} = \mathcal{B}(n, kn, p)$.
In the following lemma we consider MP on an arbitrary random graph with a fixed vertex set and a random edge set. 

\begin{lemma}\label{lemma:simple}
Let $\mathcal{E}$ be any event of MP on a random graph $\mathcal{B}$. If we have
\begin{align}\label{eq:PP_BR}
\PP_{\mathcal{B},\mathcal{Q}} \left( \mathcal{E} \right) \le \eta^2,
\end{align}
then with probability at least $1-\eta$, a random graph $B \sim \mathcal{B}$ satisfies 
\begin{align}\label{eq:PP_R-cond-B}
\PP_{\mathcal{Q}} \left( \mathcal{E} \mid B \right) \le \eta.
\end{align}
\end{lemma}

The benefits of \cref{lemma:simple} is that, 
to prove \cref{eq:PP_R-cond-B} which could be challenging to show directly, we only need to show \cref{eq:PP_BR} by constructing a revealing procedure for both $\mathcal{B}$ and $\mathcal{Q}$ simultaneously.

\begin{proof}[Proof of \cref{lemma:simple}]
Observe that
\[
\eta^2 \ge 
\PP_{\mathcal{B},\mathcal{Q}} \left( \mathcal{E} \right)
= \EE_{\mathcal{B},\mathcal{Q}} \left[ \mathbbm{1}_\mathcal{E} \right]
= \EE_{\mathcal{B}} \left[ \EE_{\mathcal{Q}} \left[ \mathbbm{1}_\mathcal{E} \mid B \right] \right]
= \EE_{\mathcal{B}} \left[ \PP_{\mathcal{Q}} \left( \mathcal{E} \mid B \right) \right].
\]
Therefore, Markov's inequality gives
\begin{align*}
\PP_{\mathcal{B}} \left( \PP_{\mathcal{Q}} \left( \mathcal{E} \mid B \right) \ge \eta \right)
\le \PP_{\mathcal{B}} \left( \PP_{\mathcal{Q}} \left( \mathcal{E} \mid B \right) \ge \frac{1}{\eta} \cdot \EE_{\mathcal{B}} \left[ \PP_{\mathcal{Q}} \left( \mathcal{E} \mid B \right) \right] \right) 
\le \eta,
\end{align*}
as claimed.
\end{proof}

\subsection{Analysis of the burn-in phase}

We now describe the first steps taken by the MP on our random construction. Crucially, our results do not depend at all on the fugacity and hence hold uniformly for any schedule, even adaptive ones. 
For the rest of this section, we fix the parameters $n,\ell,k,$ and $p$ satisfying $\eqref{eq:relations}$, and consider our random graph $\mathcal{G}(n,\ell,k,p)$, as well as the continuous-time process, as described by \cref{lem:generalMP} on the base graph $\mathcal{B}(n,kn,p)$. 

\begin{lemma}[Burn-in Phase]	
\label{lem:burn-in}
Consider the revealing procedure given in \cref{alg:burn-in} with 
\[
T = T_{\textsc{burn}} := \frac{1}{8kpn}.
\] 
Then, for every schedule $\{\lambda_t\}_{t \in [0,T]}$, with probability $1-e^{-\Omega\left(\frac{1}{p}\right)}$, 
all of the followings hold.
\begin{enumerate}
\item There exists $R_0 \subseteq R$ with $|R_0| \ge (k-1)n$ and $R_0 \cap I_T = \emptyset$.
\item There exists $L_1 \subseteq L$ with $|L_1| \ge \frac{n}{10}$ and $L_1 \subseteq I_T.$
\item The adjacency between $L_1$ and $R \setminus R_0$ is revealed and there is no edge between them. 
\item The adjacency between $L_1$ and $R_0$ is completely unrevealed and is independent of all revealed information. 
\end{enumerate}

\end{lemma}

We now follow the revealing procedure up to time $T=T_{\textsc{burn}}$ and prove each item of \cref{lem:burn-in} separately.

\begin{algorithm}[t]
	\caption{Revealing Procedure in the Burn-in Phase}\label{alg:burn-in}
	\begin{algorithmic}[1]
		\Input $n,k,T$
		\Output Part of $B \sim \mathcal{B}(n,kn,p)$ a random graph and part of $Q \sim \mathcal{Q}_T$ randomness of MP up to time $T$ 
		\Comment{Note: $B$ and $Q$ are independent}
		
		\State Reveal the records of rings of Poisson clocks at all vertices in $R$ up to time $T$. 
		
		\noindent
		Let $R_{\mathrm{bad}} \subseteq R$ be the subset of those whose clock rings at least once before time $T$, and let $R_0 = R \setminus R_{\mathrm{bad}}$. 
		\Comment{Observe that $v \notin I_T$ for all $v \in R_0$}
		
		\State Reveal the adjacency lists of all vertices in $R_{\mathrm{bad}}$.
		
		\noindent
		Let $L_{\mathrm{bad}} \subseteq L$ be the subset of vertices adjacent to $R_{\mathrm{bad}}$, and let $L_{\mathrm{good}} = L \setminus L_{\mathrm{bad}}$.
		
		\State Reveal the records of Poisson clocks, and associated $[0,1]$-uniform random variables, at all vertices in $L_{\mathrm{good}}$ up to time $T$. 
		
		\noindent
		Let $L_1 \subseteq L_{\mathrm{good}}$ be a specific subset of vertices (see \cref{lem:L1} for definition) that are occupied at time $T$. \qquad
		\Comment{Observe that the neighbors of each $u \in L_{\mathrm{good}}$ are contained in $R_0$ and are unoccupied throughout the burn-in phase.}
		
		
	\end{algorithmic}
\end{algorithm}

\begin{lemma}
Let $R_{\mathrm{bad}}$ be the set of all vertices in $R$ that are updated at least once before time $T=T_{\textsc{burn}}$. 
Define the event
\[
\mathcal{A}_1 := \left\{ |R_{\mathrm{bad}}| \le \frac{1}{4p} \right\}.
\]
Then we have $1-\PP(\mathcal{A}_1) =e^{-\Omega\left(\frac{1}{p}\right)}$.
\end{lemma}
\begin{proof}
For each vertex $v \in R$, let $N_T(v)$ be the number of times that $v$ is updated before time $T$, i.e.~the number of rings of the Poisson clock at $v$ during the time interval $[0,T]$. Thus, $N_T(v)$ is a Poisson random variable with mean $T$ and we have
\[
\Pr\left( N_T(v) \ge 1 \right) = 1 - e^{-T} \le T.
\] 
Since all $N_T(v)$'s are jointly independent for $v \in R$, we deduce from the Chernoff bound that
$$\PP\left(|R_{\mathrm{bad}}| \ge \frac{1}{4p}\right)
=\PP\left( \sum_{v \in R} \mathbbm{1}_{\{N_T(v) \ge 1\}} \ge \frac{1}{4p} \right)
    = e^{-\Omega\left(\frac{1}{p}\right)},$$
where we notice $kn \cdot T = 1/(8p) = \omega(1)$.
The statement now follows.
\end{proof}

Let $R_0 = R \setminus R_{\mathrm{bad}}$ and we have $I_t \cap R_0 = \emptyset$ for all $t \le T_{\textsc{burn}}$.

\begin{lemma}
Let $L_{\mathrm{bad}}\subseteq L$ be the set of all vertices adjacent to $R_{\mathrm{bad}}$, i.e., 
\[
L_{\mathrm{bad}} = \{u \in L: \text{$u$ is adjacent to $R_{\mathrm{bad}}$}\}.
\]
Define the event
\[
\mathcal{A}_2 := \left\{ |L_{\mathrm{bad}}| \le \frac{n}{2} \right\}.
\]
Then we have $1-\PP(\mathcal{A}_2 \mid \mathcal{A}_1) = e^{-\Omega(n)}$.
\end{lemma}
\begin{proof}
	For each $u \in L$, let $D(u)$ be the number of vertices in $R_{\mathrm{bad}}$ that are adjacent to $u$. Thus,
 conditional on $\mathcal{A}_1 = \{ |R_{\mathrm{bad}}| \le \frac{1}{4p} \}$ we have
 \[
 \Pr\left( D(u) \ge 1 \mid \mathcal{A}_1 \right) 
 \le 1 - (1-p)^{\frac{1}{4p}} 
 \le 1 - \left( 1-p \cdot \frac{1}{4p} \right) = \frac{1}{4},
 \]
 where the second inequality holds for all $p < 0.1$.
 Since the random variables $D(u)$ are jointly independent for a given $R_{\mathrm{bad}}$,
	by Chernoff's inequality, we have

\[
\PP\left( |L_{\mathrm{bad}}| \ge \frac{n}{2} \;\Big\vert\; \mathcal{A}_1 \right) 
= \PP\left( \sum_{u\in L} \mathbbm{1}_{\{D(u) \ge 1\}} \ge \frac{n}{2}  \;\Bigg\vert\; \mathcal{A}_1 \right)
= e^{-\Omega(n)},
\]
as claimed.
\end{proof}

Let $L_{\mathrm{good}} = L \setminus L_{\mathrm{bad}}$, so $|L_{\mathrm{good}}| \ge n/2$ conditional on $\mathcal{A}_1 \cap \mathcal{A}_2$. Reveal the MP restricted on $L_{\mathrm{good}}$ for $t \le T_{\textsc{burn}}$. This is exactly MP on an isolated independent set since all neighbors of $L_{\mathrm{good}}$ are contained in $R_0$, which are always unoccupied before time $T_{\textsc{burn}}$. 

\begin{lemma}\label{lem:L1}
Let $L_1$ be the set of all vertices in $L_{\mathrm{good}}$ which are updated at least once before time $T$ and that, in the last update, the associated $[0,1]$-uniform random variable $\zeta$ satisfies $\zeta > \frac{1}{3}$.
Then we have
\[
L_1 \subseteq I_T.
\]
Define the event
\[
\mathcal{A}_3 := \left\{ |L_1| \ge \frac{n}{10} \right\}.
\]
Then we have $1 - \PP(\mathcal{A}_3 \mid \mathcal{A}_1 \cap \mathcal{A}_2) = e^{-\Omega(n)}$.
\end{lemma}
\begin{proof}
Consider an arbitrary vertex $u \in L_{\mathrm{good}}$. 
Note that by the definition of $L_{\mathrm{good}}$, all neighbors of $u$ are contained in $R_0$ and thus are always unoccupied throughout the time interval $[0,T]$.

First, we observe that the probability that $u$ is updated at least once before time $T$ is 
\[
1 - e^{-\ell T} = 1 - e^{-\frac{\ell}{8kpn}} \ge 1 - e^{-1} \ge \frac{1}{2},
\]
where we use $\ell \ge 8kpn$ from \cref{eq:relations}.

Second, conditioned on $u$ being updated at least once before time $T$, 
the probability of $\zeta > \frac{1}{3}$ is $2/3$ where $\zeta$ is the associated $[0,1]$-uniform random variable in the last time when $u$ is updated (i.e.~when the Poisson clock at $u$ rings).
Note that when this happens it must hold $u \in I_T$.
To see this, observe that if $u$ is unoccupied right before the last update then $u \in I_T$, and if $u$ occupied before the last update then $\frac{1}{\ell \lambda_t} \le 1/3 < \zeta$ implies that $u$ will not be removed.
This justifies $L_1 \subseteq I_T$.

Combining the arguments above, we conclude that
\[
\Pr\left( u \in L_1 \mid u \in L_{\mathrm{good}} \right) \ge \frac{1}{3}.
\]
Note that given $L_{\mathrm{good}}$, all $\mathbbm{1}_{\{u \in L_1\}}$'s are jointly independent for $u \in L_{\mathrm{good}}$.
Therefore, an application of Chernoff bounds yields 
\[
\PP\left( |L_1| \le \frac{n}{10} \;\Big\vert\; \mathcal{A}_1 \cap \mathcal{A}_2 \right) 
=\PP\left( \sum_{u \in L_{\mathrm{good}}} \mathbbm{1}_{\{u \in L_1\}} \le \frac{n}{10} \;\Bigg\vert\; \mathcal{A}_1 \cap \mathcal{A}_2 \right)
= e^{-\Omega(n)}.
\]
This establishes the lemma.
\end{proof}

After the burn-in phase, we reveal the subgraph $B'$ induced on $W := L_1 \cup R_0$, and consider MP restricted on $B'$ with left-occupied and right-unoccupied initialization, i.e. with the independent set $L_1$ as the initialization.

\subsection[Lower bounds for the Metropolis process]{Lower bounds for the Metropolis process: Proof of \cref{thm:sparse}}


In this subsection, we prove \cref{thm:sparse}. 
We consider the case of standard Metropolis process on $G$ with uniform fugacity, i.e., $\lambda_t = \lambda$ for some fixed $\lambda \ge 1$.
Note that, however, as in \cref{lem:generalMP} the corresponding process on $B$ has different fugacities for different vertices.
We show that MP fails to find a large independent set of $G$ under any fixed fugacity $\lambda \ge 1$. 
In particular, \cref{thm:MP-all} below recovers \cref{thm:sparse} from the introduction.

\begin{theorem}[Main result for MP] 
\label{thm:MP-all}
For every increasing sequence $D=D_N$ such that $D = \Omega(1)$ and $D = o(\frac{N}{\log^2 N})$, there exists a sequence of graphs $\{G_N\}$ satisfying:
\begin{itemize}
\item $G_N$ has $\Theta(N)$ vertices, average degree $\Theta(D)$, and independence number $\alpha(G_N) = \Omega(\frac{N}{\log D})$;
\item For any fugacity $\lambda \ge 1$, the largest independent set that the Metropolis process can find within $\exp(O(\frac{N}{D}))$ steps has size $\Theta(\frac{N}{D})$, with probability at least $1-\exp(-\Omega(\frac{N}{D\log D}))$.
\end{itemize}
\end{theorem}

To establish \cref{thm:MP-all} we need \cref{lem:burn-in} for the burn-in phase and also the following lemma for steps after burn-in.
Recall that $B'$ is the subgraph induced on $W = L_1 \cup R_0$ which has not been revealed yet in the burn-in phase. 
After burn-in, we reveal the subgraph $B'$ (so that the whole graph $B$ is now fully revealed) and consider MP on $B$ with a special type of initialization (i.e., $I_0 \cap W = L_1$), as described by the events in \cref{lem:burn-in}. 
We show that, conditioned on \cref{lem:burn-in}, such MP cannot find a large independent set with high probability over the randomness of the process and over the choices of the random subgraph $B'$.
\begin{lemma}[After Burn-in]
\label{lem:after-burn-in}
Suppose the events in \cref{lem:burn-in} all hold.
The random subgraph $B'$ satisfies the following property with probability $1-e^{-\Omega(n)}$.

Consider the Metropolis process on $B$ with an initialization $I_0$ satisfying $I_0 \cap W = L_1$.
Let $T=e^{O(n)}$ and define the event 
\[
\mathcal{E} := \left\{ \forall t \in [0,T]:\, |I_t \cap R_0| \le \frac{n}{20} \right\}.
\]
Then it holds $1-\PP(\mathcal{E}) = e^{-\Omega(n)}$.
\end{lemma}


Combining \cref{lem:burn-in,lem:after-burn-in}, we are able to show \cref{thm:MP-all}.
\begin{proof}[Proof of \cref{thm:MP-all}]
Pick $n = \frac{N}{D}$, $\ell = D$, $k = \frac{D}{500\log D}$, and $p = \frac{50D\log D}{N}$ which satisfies \cref{eq:relations}. Then our construction of the random graph $\mathcal{G}(n,\ell,k,p)$ satisfies the first requirement of the theorem by \cref{fct:fact} with high probability. 
Furthermore, by \cref{lem:burn-in,lem:after-burn-in} the MP on a random bipartite graph $B \sim \mathcal{B}(n, kn, p)$ satisfies all the events in these lemmas with high probability. In particular, since $|I_t \cap L| \le |L| = n$ trivially, $|I_t \cap R_{\mathrm{bad}}| \le |R_{\mathrm{bad}}| \le n$ by \cref{lem:burn-in}, and $|I_t \cap R_0| \le \frac{n}{20}$ by \cref{lem:after-burn-in}, we have $|I_t| \le 3n = O(\frac{N}{D})$ for all $t \in [0,T]$ for the MP on $B$ with high probability.
Meanwhile, after the burn-in phase we already find an independent set of size $|I_{T_0}| \ge |L_1| \ge \frac{n}{10} = \Omega(\frac{N}{D})$ with high probability by \cref{lem:burn-in}, where $T_0 = T_{\textsc{burn}}$. 
Hence, we deduce from \cref{lemma:simple,lem:generalMP} that the random graph $\mathcal{G}(n,\ell,k,p)$ resulted from $B \sim \mathcal{B}(n, kn, p)$ satisfies the second requirement of the theorem regarding MP with high probability.
\end{proof}

The rest of this subsection is devoted to the proof of \cref{lem:after-burn-in}.

We first show that
it suffices to consider an idealized MP solely on $B'$ (i.e., neglecting all vertices and edges outside of $B'$),
and initialized from the stationary distribution conditioned on $|I \cap R_0| \le n/20$. 
Namely, we consider the MP $\{J_t\}$ on the subgraph $B' = (W,E(W))$ where $W = L_1 \cup R_0$, with the initial independent set $J_0$ generated from $\mu_{B',\lambda}(\cdot \mid |J_0 \cap R_0| \le n/20)$.

Define a partial ordering $\trianglerighteq$ such that for two independent sets $I',J'$ in $B'$, 
\begin{align*}
    I' \trianglerighteq J' 
    \qquad\text{iff}\qquad 
    I' \cap L_1 \supseteq J' \cap L_1
    ~\text{and}~
    I' \cap R_0 \subseteq J' \cap R_0. 
\end{align*}
Thus, $L_1$ is the unique maximum independent set under the ordering $\trianglerighteq$.

\begin{lemma}\label{lem:monotone}
There is a coupling between the Metropolis process $\{I_t\}$ on $B$ with $I_0 \cap W = L_1$ and the idealized Metropolis process $\{J_t\}$ on $B'$ with initialization $J_0 \sim \mu_{B',\lambda}(\cdot \mid |J_0 \cap R_0| \le n/20)$ such that 
\[
I_t \cap W \trianglerighteq J_t
\]
for all $t \in [0,T]$ almost surely.
\end{lemma}

\begin{proof}
We prove the lemma by constructing a coupling $\{(I_t, J_t)\}$ of the two processes such that almost surely $I_t \cap W \trianglerighteq J_t$ for all $t \in [0,T]$.
It is helpful to modify the description of the continuous-time Metropolis process in the following equivalent way: We place at each vertex $v \in V$ a Poisson clock of rate $2\rho(v)$, and whenever the clock rings, with probability $1/2$ we do nothing (discard the ring) and with the remaining probability $1/2$ we update the vertex in the usual way (pick a real in $[0,1]$ u.a.r.\ and then use \cref{alg:MP-1step}). This modified version can be understood as the lazy version of the Metropolis process.

We couple all the Poisson clocks in $L_1 \cup R_0$ identically for both $\{I_t\}$ and $\{J_t\}$ and then couple the updates inside $L_1 \cup R_0$ in a careful way which we shall specify next. 
For updates for $\{I_t\}$ outside of $L_1 \cup R_0$ we do not care and they can be arbitrary.
Almost surely the total number of clock rings at all vertices in $L_1 \cup R_0$ is finite and all these rings happen at distinct time in $(0,T)$. 
In the rest of the proof we assume this holds and denote the times of all clock rings by $0< t_1<\dots < t_m < T$ where $m \in \N$ is the total number of rings. 
Further define $t_0=0$.
We construct a coupling by induction, conditional on the sequence of update times.
Initially, 
since $I_0 \cap W = L_1$ is maximum, we have $I_0 \cap W \trianglerighteq J_0$.

Assume $I_{t_{i-1}} \cap W \trianglerighteq J_{t_{i-1}}$ for some $1\le i \le m$ and at time $t_i$ the Poisson clock rings at some vertex in $L_1 \cup R_0$.

\medskip
\textit{Case 1: The ring happens at some $u \in L_1$.}
\begin{itemize}
    \item[(1a)] If $u \in I_{t_{i-1}}$ and $u \in J_{t_{i-1}}$, then we can couple the updates at $u$ identically in both processes so that $u \in I_{t_i}$ iff $u \in J_{t_i}$. 
    
    \item[(1b)] If $u \in I_{t_{i-1}}$ and $u \notin J_{t_{i-1}}$, then we couple as follows: When $I_{t_i}$ makes an update at $u$ (calling \cref{alg:MP-1step}), $J_{t_i}$ discards the ring (no update); and when $I_{t_i}$ discards the update, $J_{t_i}$ makes the update at $v$. 
Hence, $I_{t_i} \cap W \trianglerighteq J_{t_i}$ always.

\item[(1c)] If $u \notin I_{t_{i-1}}$ and $u \notin J_{t_{i-1}}$, then we couple as follows: With probability $1/2$ we discard the rings in both processes (no updates), and with the remaining probability $1/2$ we update $u$ in both processes by including $u$ into the current independent set if possible. 
We claim that $I_{t_i} \cap W \trianglerighteq J_{t_i}$ after the update.
Note that the only possible way of violating this is that $u \notin I_{t_i}$ and $u \in J_{t_i}$. 
This means that there exists a neighbor $v$ of $u$ satisfying $v \in I_{t_{i-1}}$ so that $u$ is blocked by $v$ for this update. 
Since $u \in L_1$, it holds $v \in R_0$ by \cref{lem:burn-in}. 
Therefore, combining $I_{t_{i-1}} \cap W \trianglerighteq J_{t_{i-1}}$, $v \in I_{t_{i-1}}$, and $v \in R_0$, we must have $v \in J_{t_{i-1}}$ as well. 
This implies that $u$ should also be blocked and cannot be included in $J_{t_i}$, contradicting to $u \in J_{t_i}$. 
\end{itemize}

\medskip
\textit{Case 2: The ring happens at some $v \in R_0$.}
\begin{itemize}
    \item[(2a)] If $v \in I_{t_{i-1}}$ and $v \in J_{t_{i-1}}$, then we can couple the updates at $v$ identically in both processes so that $v \in I_{t_i}$ iff $v \in J_{t_i}$. 
    \item[(2b)] If $v \notin I_{t_{i-1}}$ and $v \in J_{t_{i-1}}$, then we couple in the same way as (1b): When $I_{t_i}$ makes the update, $J_{t_i}$ discards it; and when $I_{t_i}$ discards the update, $J_{t_i}$ makes it. 
Hence, $I_{t_i} \cap W \trianglerighteq J_{t_i}$ always.

\item[(2c)] If $v \notin I_{t_{i-1}}$ and $v \notin J_{t_{i-1}}$, then similarly as (1c) we couple as follows: With probability $1/2$ we discard the rings in both processes, and with the remaining probability $1/2$ we update $v$ in both processes by including $v$ if possible. 
We claim that $I_{t_i} \cap W \trianglerighteq J_{t_i}$.
The only possible way of violating this is that $v \in I_{t_i}$ and $v \notin J_{t_i}$. 
This means that some neighbor $u$ of $v$ satisfies $u \in J_{t_{i-1}}$ so that $v$ is blocked in this update by $u$. 
And we must have $u \in L_1$ since the process $\{J_t\}$ is defined only on $B'$.
Therefore, combining $I_{t_{i-1}} \cap W \trianglerighteq J_{t_{i-1}}$, $u \in J_{t_{i-1}}$, and $u \in L_1$, we must have $u \in I_{t_{i-1}}$ as well, 
contradicting to $v \in I_{t_i}$. 
\end{itemize}

Hence, we have $I_{t_i} \cap W \trianglerighteq J_{t_i}$ in all cases. 

By induction, we conclude that $I_{t_i} \cap W \trianglerighteq J_{t_i}$ for all $i=0,1,\dots,m$ conditional on the update sequence $0=t_0 < t_1 < \dots < t_m < T$. 
Therefore, almost surely we have $I_t \cap W \trianglerighteq J_t$ for all $t \in [0,T]$.
%
%
\end{proof}

By \cref{lem:monotone}, we only need to consider the idealized MP $\{J_t\}$ on $B'$.

\begin{lemma}\label{lem:bottleneck}
The random graph $B'$ satisfies the following property with probability $1-e^{-\Omega(n)}$.

Consider the idealized MP $\{J_t\}$ on $B'$ with an initialization $J_0$ generated from $\mu_{B',\lambda}(\cdot \mid |J_0 \cap R_0| \le n/20)$. 
Let $T=e^{O(n)}$ and define the event 
\[
\mathcal{E} := \left\{ \forall t \in [0,T]:\, |J_t \cap R_0| \le \frac{n}{20} \right\}.
\]
Then it holds $1-\PP(\mathcal{E}) = e^{-\Omega(n)}$.
\end{lemma}

\begin{proof}
Define the following graph property $\mathcal{S}$:
there is no independent set $I$ such that $|I \cap L_1| = n/20$ and $|I \cap R_0| = n/20$.

We first prove that the random bipartite graph $B'$ satisfies $\mathcal{S}$ with probability $1-e^{-\Omega(n)}$.
Let $Z$ be the number of such independent sets; namely those with $|I \cap L_1| = n/20$ and $|I \cap R_0| = n/20$.
For a random $B'$, we have
\begin{align*}
\E[Z] \le \binom{n/10}{n/20} \cdot \binom{kn}{n/20} \cdot \left( 1-p \right)^{n/20 \cdot n/20}. 
\end{align*}
The product of the first two terms is upper bounded by
\begin{align*}
\binom{n/10}{n/20} \cdot \binom{kn}{n/20}
\le 2^{n/10} \cdot \left( 20e k \right)^{n/20} 
\le \left( 240 k \right)^{n/20} \le k^{n/10}. 
\end{align*}
Therefore, by \cref{eq:relations} we have
\begin{align*}
\E[Z] 
\le \exp\left( - \left( \frac{pn}{400} - \frac{\log k}{10} \right) n \right)
= e^{-\Omega(n)}.
\end{align*}
Thus, $\PP(Z\ge 1) \le \E[Z] = e^{-\Omega(n)}$ by Markov's inequality.

Suppose $\mathcal{S}$ holds.
We can then apply a simple conductance argument.
Let $\mathcal{I}_{\le n/20}$ be the collection of all independent sets $I$ such that $|I \cap R_0| \le n/20$.  
Let $\mathcal{I}_{= n/20}$ be the collection of all independent sets $I$ such that $|I \cap R_0| = n/20$.  
Assuming $\ell \ge 300 k$ from \cref{eq:relations}, we deduce that
\begin{align*}
\frac{\mu_{B',\lambda}(\mathcal{I}_{=n/20})}{\mu_{B',\lambda}(\mathcal{I}_{\le n/20})}
\le \frac{\binom{n/10}{\le n/20} \cdot \binom{kn}{n/20} \cdot (\ell \lambda)^{n/20} \cdot \lambda^{n/20}}{(\ell \lambda)^{n/10}} 
\le \frac{\left( 240 k \right)^{n/20}}{\ell^{n/20}}
= e^{-\Omega(n)}.
\end{align*}
The lemma then follows from, e.g., Theorem 7.4 of \cite{levin2017markov} (more specifically, Eq.~(7.10) in its proof) or \cite{dyer2002counting,mossel2009hardness}.
\end{proof}

We are now ready to prove \cref{lem:after-burn-in}.
\begin{proof}[Proof of \cref{lem:after-burn-in}]
Consider the idealized MP $\{J_t\}$ on $B'$ as defined before \cref{lem:monotone}. Since by \cref{lem:monotone} we have $I_t \cap W \trianglerighteq J_t$ for all $t \in [0,T]$ almost surely where $T = e^{O(n)}$, we know that if $J_t$ satisfies $\mathcal{E}$ then it must hold $I_t$ satisfies $\mathcal{E}$. 
The probability of the former is bounded by \cref{lem:bottleneck}, which gives
\begin{align*}
\PP_{\{I_t\}}(\mathcal{E}) 
\ge \PP_{\{J_t\}}(\mathcal{E}) = 1-e^{-\Omega(n)},
\end{align*}
with high probability over the choices of the random graph $B'$.
The lemma then follows.
\end{proof}

\subsection[Simulated Annealing on dense graphs]{Simulated annealing on dense graphs: Proof of \cref{thm:dense}}
We now fix the parameters in the definition of $\mathcal{G}(n,\ell,k,p)$ as follows. It will be convenient to slightly re-parameterize the random graph. For two small constants $\eps$ and $\delta$ satisfying $0 < \frac{\eps}{4}< \delta < \eps< \frac{1}{3}$, let $m$ be such that $m^{\eps} = n$ (for simplicity we shall treat $m$ as an integer) and set 
$$p = m^{-\delta},k = m^{1-3\eps}, \ell = m^{1-\eps}.$$
 By noting $n^{\frac{1}{\eps}} = m$ it is readily seen that the conditions in \eqref{eq:relations} are met. Observe that with these parameters, a graph sampled from $\mathcal{G}(m^\eps,m^{1-\eps},m^{1-3\eps},m^{-\delta})$ will have an independent set of size $m^{1-\eps}$, consisting of the one side of the bi-bipartite base graph, yet any independent set which is mostly supported on the other side can only have size at most $m^\eps$.

 We shall now show that under the events defined by the burn-in phase, no new vertices from $R$ can be added to the independent set. This fact shall constitute the main hurdle for MP to find a large independent set.
 \begin{lemma} \label{lem:postdense}
 	Let the notation of \cref{lem:burn-in} prevail and define the stopping time,
 	$$\tau :=\min\{t \geq T_{\textsc{burn}}|I_t \cap R_0 \neq \emptyset\}.$$
 	Then, for \emph{any} fugacity schedule $\{\lambda_t\}_{t \geq 0}$, which can be adaptive,
 	$$\PP\left(\tau > e^{\Omega(m^{\frac{\eps}{4}})}\right) \leq e^{-\Omega(m^{\frac{\eps}{4}})}.$$
 \end{lemma}
\begin{proof}
	We perform all calculations conditional on the event defined by \cref{lem:burn-in}. Similar to the proof of \cref{lem:monotone} we enumerate the times $\{t_i\}_{i=1}^T$ where $T_{\textsc{burn}} \leq t_1 <\dots < t_T \leq  e^{\Omega(m^{\frac{\eps}{4}})}$. By \cref{lem:generalMP}, the number of rings, $T$, is polynomially related to the time horizon,  so it will be enough to consider the discrete-time MP restricted to $L_1$, starting from $T_{\textsc{burn}}$. That is, from now on $I_t$ will stand for $I_{t_t + T_{\textsc{burn}}}\cap L_1$. We enumerate $R_0 = \{v_i\}_{i=1}^M$ where $M \geq m^{1-\eps}-m^{\delta/2}$ and for every $i = 1,\dots,M$ define the random process 
	$$\Gamma^i_t = \#\{u \in I_t| u \sim v_i\}.$$
	We define the auxiliary stopping time $\tau' := \min\{t| \exists i = 1,\dots,M \text{ such that } \Gamma_i = 0\}$ and observe that for any $t \geq 0$,
	\begin{equation} \label{eq:stopdom}
		\PP_{B',\mathcal{Q}'}( \tau \le t ) \leq \PP_{B',\mathcal{Q}'}(\tau' \leq t-1).
	\end{equation}
	Thus, let us analyze the process $\Gamma_t^i$. Recall that, by definition, $I_0 = L_1$. So, since $|L_1| \geq \frac{1}{10}m^\eps$, by applying Chernoff's inequality, we can show,
	$$\PP\left(\forall i=1,\dots M, \ \Gamma_0^i \geq \frac{1}{100}m^{\frac{\eps}{2}} \right) = 1 - e^{\Omega(m^{\eps})}.$$
	Now, let us fix some $i =1,\dots,M$, and consider the case that for some $t \geq 0$, $\Gamma_t^i \leq m^{\frac{\eps}{4}}$. In this case, we have
	$$\PP\left(\Gamma_{t+1}^i = \Gamma_t^i +1 \right)\geq \frac{(\Gamma_0^i - \Gamma_t^i)}{m + m^{1-2\eps}}\geq \frac{1}{1000m^{1-\frac{3\eps}{2}}} \geq \frac{1}{m^{1-\eps}},$$
	and 
	$$\PP\left(\Gamma_{t+1}^i = \Gamma_t^i -1 \right)\leq \frac{1}{\lambda_t}\frac{\Gamma_t^i}{m + m^{1-2\eps}}\leq \frac{1}{m^{1-\frac{3\eps}{4}}}.$$
	Both bounds follow from counting the number of vertices which can affect a change in $\Gamma_t^i.$ For the upper bound we also factor the deletion probabilities and use that $\lambda_t \geq 1$. 
	
	Consider the biased random walk with independent increments $Y_j$ satisfying,
	$$\PP\left(Y_j = 1\right) = \frac{m^{-(1-\eps)}}{m^{-(1-\eps)} + m^{-(1-\frac{3\eps}{4})}}, \text{ and } \PP\left(Y_j = -1\right) = \frac{m^{-(1-\frac{3\eps}{4})}}{m^{-(1-\eps)} + m^{-(1-\frac{3\eps}{4})}},$$
	and denote $S_r = \sum\limits_{j=1}^r Y_j$.
	Since $Y_j$ is a biased random walk with positive drift, we can bound the probability that it ever reaches some negative value. More specifically, the results in \cite[(2.8), Chapter XIV.2]{feller1968probability} show that
	\begin{align} \label{eq:minbiasedRW}
		\PP\left(\min\limits_{r \geq 0} S_r \leq -m^{\frac{\eps}{4}}\right) \leq \left(\frac{m^{1-\eps}}{m^{1-\frac{3\eps}{4}}}\right)^{m^{\frac{\eps}{4}}}\leq m^{-\frac{\eps}{4}m^{\frac{\eps}{4}}}\leq m^{-m^{\frac{\eps}{8}}}.
	\end{align}
	Now, for any $t_0 \geq 0$ such that $\Gamma_{t_0}^i = m^{\frac{\eps}{4}}$, we define a new stopping time 
	$$\tau'(t_0) = \min \{t > t_0| \Gamma_t^i = m^{\frac{\eps}{4}} \text{ or } \Gamma_t^i = 0\}.$$
	But, for any $t \in [t_0, \tau'(t_0)]$, we have shown that $\Gamma^i_{t_0 + t}$ stochastically dominates (a lazy version of) $S_t + m^{\frac{\eps}{4}}$. So, \eqref{eq:minbiasedRW} shows,
	$$\PP\left(\Gamma^i_{\tau'(t_0)} = 0\right) \leq m^{-m^{\frac{\eps}{4}}}.$$
	Finally, since $\Gamma_0^i > m^{\frac{\eps}{4}}$, we have that $\tau' = \tau'(t')$, for some random $t'$, and in particular $\tau' > t'$. By a union bound,
	\begin{align*}
		\PP\left(\tau' \leq m^{\frac{1}{2}m^{\frac{\eps}{4}}}\right) &\leq \PP\left(\tau' = \tau(t'), \text{ for some } t' \leq m^{\frac{1}{2}n^{\frac{\eps}{4}}}\right)\\
		& = \PP\left(\Gamma^i_{\tau'(t')} = 0, \text{ for some } t' \leq m^{\frac{1}{2}m^{\frac{\eps}{4}}}\right) \leq m^{-m^{\frac{\eps}{4}}}m^{\frac{1}{2}m^{\frac{\eps}{4}}} \leq m^{-\frac{1}{2}m^{\frac{\eps}{4}}}.
	\end{align*}
	The proof concludes by taking a union bound over all $n$ and invoking \eqref{eq:stopdom}.
\end{proof}
We can now prove \cref{thm:dense}
\begin{proof}[Proof of \cref{thm:dense}]
	Since $\frac{1}{p} = m^\delta \geq m^{\frac{\eps}{4}}$, applying \cref{lem:burn-in} and \cref{lem:postdense} we see that with probability $1-e^{-\Omega(m^{\frac{\eps}{4}})}$
	$$I_t \subset L \cup (R \setminus R_0),$$ for every $t \leq e^{\Omega(m^{\frac{\eps}{4}})}$.
	Thus, by construction, and recalling $m^\eps = n$, for every such $t$, $|I_t| \leq |L| + |R \setminus R_0| \leq n + n = 2n = 2m^\eps.$
	On the other hand, by \cref{fct:fact}, a graph sampled from $\mathcal{G}(m^\eps,m^{1-\eps},m^{1-3\eps},m^{-\delta})$ has $\Theta(m)$ vertices and an independent set of size $\Theta(m^{1-2\eps})$. Switching roles between $n$ and $m$, and applying \cref{lemma:simple} finishes the proof.
\end{proof}

\section{Hardness Results for Bipartite Graphs} \label{sec:bipartite}
In this section, we focus on bipartite graphs and prove \cref{thm:bipartite}. The main technical novelty of the proof is a reduction between the Metropolis algorithm and the randomized greedy algorithm on bipartite graphs. The idea is to blow up a hard instance for randomized greedy in a way that makes the added randomness of Metropolis inefficient. 

Recall that the randomized greedy algorithm is equivalent to the Metropolis process at $0$ temperature, or alternatively $\lambda_t \equiv \infty$. In other words, the algorithm chooses random vertices uniformly at random, adds them to a growing independent set whenever possible, and never deletes them. Thus, the algorithm always terminates after each vertex is chosen at least once, which happens in finite time almost surely. For the remainder of this section, for a graph $G$, we shall denote by $I^G$, the (random) independent set obtained at the termination of the randomized greedy algorithm on $G$.

Having introduced the randomized greedy algorithm \cref{thm:bipartite} is now an immediate consequence of the following two results.
\begin{proposition} \label{prop:reduction}
    Suppose that there exists a family of bipartite graphs $\{G_n\}_{n\geq 0}$ on $\Theta(n)$ vertices, and a function $r:\N\to (0,1)$, such that,
    $$\PP\left(|I^{G_n}| > r(n)\alpha(G_n)\right) = e^{-n^\eta},$$
    for some $\eta > 0$.
    Then, for $\eps > 0$ small enough, there exists a family of bipartite graphs $\{\tilde G_n\}_{n\geq 0}$ on $\mathrm{poly}(n)$ vertices, such if $I_t$ stands for the metropolis process on $\tilde{G}_n$ with any temperature,
    $$\PP\left(\max\limits_{t < e^{n^{\eta'}}} |I_t| \geq \left(r(n) + \frac{2}{n^{1-\eps}}\right)\alpha(G_n)\right) = e^{-n^{\eta'}},$$
    for some $\eta' \leq \eta.$
\end{proposition}

\begin{proposition} \label{prop:hardgreedy}
    Let $d_n < \frac{\log(n)}{100}$ be a sequence of numbers. There exists a sequence of bipartite graphs $G_n$ on $\Theta(n)$ vertices and average degree $\Theta(d_n)$, such that 
    $$\PP\left(|I^{G_n}| > (4+o(1))\frac{\log(d_n)}{d_n}\alpha(G_n)\right) = e^{-n^\eta},$$
    for some $\eta > 0$.
\end{proposition}
\cref{prop:reduction} essentially says that hard instances for randomized greedy can be used to construct hard instances for Metropolis, with any cooling schedule, and \cref{prop:hardgreedy} asserts that hard instances for randomized greedy do exist. \cref{thm:bipartite} immediately follows by coupling these two facts and appropriately choosing $d_n$. Thus, the rest of this section is devoted to the proofs of the two propositions. In \cref{sec:reduction} below we prove \cref{prop:reduction} and in \cref{sec:greedy} we prove \cref{prop:hardgreedy}.

\subsection{From randomized greedy to Metropolis} \label{sec:reduction}
We begin by explaining how to blow-up hard instances for randomized greedy into hard instances for Metropolis. Let $G$ be a graph on $n$ vertices $\{v_i\}_{i=1}^n$, and denote by $I^G$ the (random) independent set obtained by running randomized greedy on $G$. Assume that, for some $r(n) > 0$, and $\eta > 0$,
\begin{equation} \label{eq:baseass}
    \PP\left(|I^G| > r(n)\alpha(G)\right) \leq e^{-n^\eta}.
\end{equation}
Our goal is to show that there exists a graph for which similar estimates hold when we replace randomized greedy with the metropolis process with polynomially many steps.

Let us describe the hard instance. For $K,M > 10$, define the the graph $G^{K,M}$ in the following way:
\begin{enumerate}
    \item First, let $G^M$ be the disjoint union of $M$ copies of $G$.
    \item $G^{K,M}$  is obtained as the $K$ blow-up of $G^M$. That is, every node is replaced by an independent set of size $K$, and a complete bipartite graph replaces every edge.
\end{enumerate}
As usual, we shall write $I_t$ for the set maintained by MP on $G$ and enumerate the vertices of $G^{K,M}$ as $v_{i,m,k}$, the $k^{th}$ element in the blow-up of $v_i$ in the $ m^{th}$ copy of $G$. 
We shall refer to the set of  vertices $v_{i,m} = \{v_{i,m,k}\}_{k=1}^K$ corresponding to a vertex $v_i$ and in the $m^{th}$ copy,
as a \emph{cloud}. Clearly, every cloud has exactly $K$ vertices.
The following quantity shall play a central role,
$$v_{i,m}^t := |I_t \cap \{v_{i,m,1},\dots,v_{i,m,K}\}|,$$
the load in the cloud $v_{i,m}$, at time $t$. The idea is that once a cloud becomes occupied, MP is more likely to add more vertices from the cloud than to remove existing ones. Thus it is very unlikely that a cloud will become empty once occupied. This should be seen as an analogy to randomized greed; as long as no cloud has emptied one can simulate randomized greedy on a given component of the graph. Having many different copies ensures that on most copies no cloud will empty.

In light of this, we first show that once the load becomes positive, and so a cloud becomes occupied, it is very unlikely to drop to $0$ again.
\begin{lemma} \label{lem:nodeloading}
For any $t_0 \geq 0$, and $t_0<t <K^{\frac{K^{1/3}}{2}}$ it holds that,
$$P\left(v_{i,m}^{t+t_0} = 0 |v^{t_0}_{i,m} = 1\right) \leq \frac{C}{K^{1/3}},$$
for some absolute constant $C>0.$
\end{lemma}
\begin{proof}
    Let $t_0' > t_0 + K^{\frac{1}{3}}$ be the first time that $K^{\frac{1}{3}}$ where chosen from the cloud $v_{i,m}$,
    and observe  
    \begin{equation} \label{eq:bins}
        \PP\left(\exists t \in [t_0, t_0']: v_{i,m}^{t} = 0|v_{i,m}^{t_0}=1\right)\geq \PP\left(v_{i,m}^{t_0'} = K^{\frac{1}{3}}|v_{i,m}^{t_0}=1\right)\geq \frac{1}{K^{\frac{1}{3}}}.
    \end{equation}
    Indeed, this estimate follows from a standard Balls and Bins argument. The probability that after the first $K^{\frac{1}{3}}$
    choices of vertices in $v_{i,m}$ the algorithm deletes a vertex from $I_t$ is at most 
    $$\sum_{i=1}^{K^{\frac{1}{3}}}\frac{i}{K} \geq \frac{K^{\frac{2}{3}}}{K} \geq \frac{1}{K^{\frac{1}{3}}},$$ since by the union bound this upper bounds the probability $K^{\frac{1}{3}}$ balls will have at least one collision when distributed randomly across $K$ bins. This estimate gives the right inequality in \eqref{eq:bins}. The left inequality in \eqref{eq:bins} then follows from the fact that to reach a load $K^{\frac{1}{3}}$ every chosen vertex, out of the $K^{\frac{1}{3}}$ was added to the cloud, and so no vertex was removed.

    Conditional on the event $v_{i,m}^{t_0'} = K^{\frac{1}{3}}$ denote now $Y_t = v_{i,m}^{t+t'_0}$ and observe that as long as $Y_t< 2K^{\frac{1}{3}}$ we have that, at any $t$ in which the value $Y_t$ changes,
    $$P\left(Y_t  - Y_{t-1} = 1\right) \geq \frac{K - 2K^{\frac{1}{3}}}{K},$$
    $$P\left(Y_t  - Y_{t-1} = -1\right) \leq \frac{2K^{\frac{1}{3}}}{K}.$$
    Define the stopping time $\tau = \min\{t| Y_t = 2K^{\frac{1}{3}} \text{ or } Y_t = 0\}$. Similar to the proof of \cref{thm:dense}, since $Y_t$ is stochastically dominated by a biased random walk with the above increments and since $Y_0 = K^{\frac{1}{3}}$, the results in \cite[(2.8), Chapter XIV.2]{feller1968probability} imply 
    $$\PP\left(Y_\tau = 0\right) \leq \left(\frac{2K^{\frac{1}{3}}}{K-2K^{\frac{1}{3}}}\right)^{K^{\frac{1}{3}}}\leq \left(4\frac{K^{\frac{1}{3}}}{K}\right)^{K^\frac{1}{3}} \leq  4K^{-\frac{2K^{\frac{1}{3}}}{3}}.$$
    The proof concludes by iterating this argument for a polynomial number of steps and applying a union bound.
    \end{proof}

For fixed time $t > 0$ we now define the number of deloaded clouds as 
    $$\mathrm{deload}(t):= \#\{v_{i,m}| \exists t'<t''\leq t \text{ such that } v_{i,m}^{t'} = 1 \text{ and } v_{i,m}^{t''}=0\}.$$   
In words, $\mathrm{deload}(t)$ measures the number of clouds that were at some point occupied (had at least one vertex chosen by the algorithm) and later become unoccupied (all vertices in the cloud were deleted at a later point). The main upshot of the previous result is that the number of deloaded vertices remains very small after polynomially many iterations.
\begin{lemma} \label{lem:deloadingnumber}
  Suppose that $C\frac{nM}{K^{\frac{1}{3}}} \geq 1$. Then, for any $t <K^{\frac{K^{1/3}}{2}}$, and $\eps > 0$
    $$\PP\left(\mathrm{deload}(t) \geq n^\eps \frac{nM}{K^{\frac{1}{3}}}\right) = e^{-\Omega(n^\eps)}.$$
\end{lemma}
\begin{proof}
    Observe that $G^{K,M}$ contains $nM$ clouds. Thus, for $t <K^{\frac{K^{1/3}}{2}}$, by \cref{lem:nodeloading}, $\mathrm{deload}(t)$ is stochastically dominated by $B:=\mathrm{binomial}\left(nM, \frac{C}{K^{\frac{1}{3}}}\right)$. Since $\EE[B] = C\frac{nM}{K^{\frac{1}{3}}} \geq 1$, by Chernoff's inequality,
    $$\PP\left(B \geq n^\eps \frac{nM}{K^{\frac{1}{3}}}\right) \leq e^{-\Omega(n^\eps)}.$$
    The proof is complete.
\end{proof}
We can now prove \cref{prop:reduction}.


\begin{proof}[Proof of \cref{prop:reduction}]
Set $M = n$ and $K = n^6$, so that $\frac{nM}{K^{\frac{1}{3}}} = 1$. 
By \cref{lem:deloadingnumber}, with probability $e^{-\Omega(n^\eps)}$ at most $n^\eps$ clouds were de-loaded by time $t \leq K^{\frac{K^{1/3}}{2}}$. With no loss of generality, they belong to the first $n^\eps$ copies in $G^M$, $\{G_m\}_{m=1}^{n^\eps}$. On the other $M - n^\eps$ copies, $\{G_m\}_{m=n^\eps}^M$, since no deloading happened we can couple Metropolis on $G_m$ with (a lazy version of) randomized greedy on the base graph $G$ in the following way:

If Metropolis chooses vertex $v_{i,m,k}$, from cloud $v_{i,m}$ we let randomized greedy choose $v_{i}$ from $G$. Since all clouds have the same size, the probability of choosing a vertex from cloud $v_{i,m}$ is equal to the probability of choosing the vertex $v_i$. Moreover, since edges exist only between clouds, whenever $v_{i,m,k}$ can be added to $I_t$ then either $v_i$ can also be added, if $v_{i,m}^t = 0$, or $v_i$ was already added, if $v_{i,m}^t > 0$. If metropolis removes $v_{i,m,k}$ then randomized greedy does nothing and maintains its chosen $v_i$ in the independent set. This coupling remains valid until the first deloading happens in $G_m$.

Thus, by the assumption \eqref{eq:baseass} on the base graph $G$, with probability $1-e^{-\Omega(n^\eps)}$ for every $m \geq n^\eps$, $\PP\left(|I_t \cap G_m| > r(n)K\alpha(G)\right) \leq e^{-n^\alpha}$, where we allow metropolis to fill in all $K$ vertices from every occupied cloud. So, suppose that $\eps < \alpha$, then 
$$\PP\left( \forall m \geq n^\eps,\, |I_t \cap G_m| < r(n)K\alpha(G)\right) \geq 1 - 2Me^{-n^\alpha}.$$
On the other hand, clearly for $m \leq n^\eps$ we have $|I_t \cap G_i| \leq Kn$. 
It follows that 
$$\PP\left(|I_t| \leq Kn\cdot n^\eps + r(n)KM\alpha(G)\right) \geq 1 - 2Me^{-n^\alpha}.$$
Finally, by construction $G^{K,M}$ has a maximal independent set of size $\alpha(G)Mk$, so the approximation ratio is $\frac{Kn\cdot n^\eps+ r(n)KM\alpha(G)}{\alpha(G)MK}$. Let us choose now $M = n$ and $K = n^2$ 
to get an approximation ratio of $\frac{n^\eps}{\alpha(G)} + r(n) \leq \frac{2}{n^{1-\eps}} + r(n)$. Here we have used the fact that if $G$ is a bipartite graph on $n$ vertices then $\alpha(G) \geq \frac{n}{2}$. 
\end{proof}

\subsection{Randomized Greedy on bipartite graphs} \label{sec:greedy}
To prove \cref{prop:hardgreedy} we will show that random bipartite graphs, with sufficiently small average degree, can serve as hard instances for randomized greedy.
Let $G = \mathcal{B}(2n,p_n)$ be a random bipartite graph on $2n$ vertices where each vertex is assigned a side independently and uniformly at random. Moreover, for some $d:=d_n$ satisfying $d_n \leq \frac{1}{100}\log(n)$ each edge appears independently with probability $p_n = \frac{d}{n}$. 
The randomized greedy algorithm on $G$ is equivalent to the following pair of stochastic processes $(L_t,R_t)$, defined as follows:
\begin{enumerate}
    \item Initialize $L_0 = R_0 = 0$ and let $V \in \{-1,1\}^{2n}$ be a drawn uniformly at random.
    \item At time $t>0$, if $V(t) = 1$ then $R_{t} = R_{t-1}$ and 
    $$\PP\left(L_{t} = L_{t-1} + 1\right) = (1-p_n)^{R_t}, \PP\left(L_{t} = L_{t-1}\right) = 1- (1-p_n)^{R_t}.$$ 
    Otherwise, $V(t) = -1$ and $L_{t} = L_{t-1}$ with
    $$\PP\left(R_{t} = R_{t-1} + 1\right) = (1-p_n)^{L_t}, \PP\left(R_{t} = R_{t-1}\right) = 1- (1-p_n)^{L_t}.$$ 
\end{enumerate}
To unpack the definition, $V$ is the random order in which we reveal the vertices of $V$, because the greedy algorithm cannot backtrack we can disregard vertices that were chosen twice, and, due to symmetry, we only care if a vertex belongs to the left or right side ($1$ or $-1$ at random). $L_t$ represents the independent set on the left side, after $t$ vertices were revealed, and $R_t$ represents the right side. When a new vertex is revealed on the right side, the left side's size remains unchanged, while the right side grows if and only if the newly revealed vertex is not connected to any vertex in the independent set on the left side. Since each edge appears with probability $p_n$ and at time $t$ the independent set on the left side is of size $L_t$, the right side grows with probability $(1-p_n)^{L_t}$. Altogether, the size of the set randomized greedy finds is of size $L_{2n} + R_{2n}.$ Moreover observe that with overwhelming probability the size of the maximal independent set is $(1+o(1))n$.

We first claim that the discrepancy between $L_{2n}$ and $R_{2n}$ must necessarily be small.
\begin{lemma} \label{lem:discrep}
It holds that
    $$\PP\left(|L_{2n}-R_{2n}| \geq n^{0.9}\right) \leq e^{-\Omega(n^{0.1})}.$$
\end{lemma}
We shall momentarily prove the lemma, but first, let us show that it implies that randomized greedy cannot find a large independent set and prove \cref{prop:hardgreedy}
\begin{proof}[Proof of \cref{prop:hardgreedy}]
 Let $G$ be a random graph as above and recall that $I^G$ stands for the independent set obtained by randomized greedy. Fix $\eps > 0$ and define the event $A:= \{|L_{2n} - R_{2n}| < n^{0.7}\}.$ We will show that 
 $$\PP\left(R_{2n} + L_{2n} \geq \frac{(4+\eps)\log(d)}{d}n|A\right) \leq e^{-\Omega(n^{0.1})}.$$
 Since $|I^G| = L_{2n} + R_{2n}$, the result will follow from the law of total probability. Indeed,
 \begin{align*}
     \PP\left(|I^G| \geq \frac{(4+\eps)\log(d)}{d}n\right)&=\PP(A)\PP\left(R_{2n} + L_{2n} \geq\frac{(4+\eps)\log(d)}{d}n|A\right) \\
     &\ \ \ \ \ \ + \PP(A^c)\PP\left(R_{2n} + L_{2n} \geq \frac{(4+\eps)\log(d)}{d}n|A^c\right)\\
     &\leq \PP\left(R_{2n} + L_{2n} \geq \frac{(4+\eps)\log(d)}{d}n|A\right) + \PP(A^c) \leq 2e^{-\Omega(n^{0.1})},
 \end{align*}
 where the inequality before last is due to \cref{lem:discrep}.
 So, to finish the proof, suppose that 
 $$|I^G| \geq \frac{(4+\eps)\log(d)}{d}n.$$
 Hence, we may assume, with no loss of generality, that $L_{2n} \geq \frac{|I^G|}{2} \geq \frac{(2+\frac{\eps}{2})\log(d)}{d}n$, and under $A$, if $n$ is large enough, we will get,
 $$R_{2n} \geq L_{2n} - n^{0.9} > \frac{(2+ \frac{\eps}{4})\log(d)}{d}n.$$
 By standard results about independent sets in random bipartite graphs (see \cref{lem:bipartiteindep} below), it cannot hold that both $R_{2n}$ and $L_{2n}$ are larger than $$\frac{(2+ \frac{\eps}{4})\log(d)}{d}n.$$ 
\end{proof}
while proved before ~\cite{clementi1999improved,perkins2023hardness}, we outline the argument to bound from above the size of balanced bipartite independent sets in random bipartite graphs next.
\begin{lemma} \label{lem:bipartiteindep}
Let $G$ be a bipartite graph with sides $L,R$ both of cardinally $(1+o(1))n$ and where every edge between $L,R$ occurs independently with probability $d/n$ with $e^2<d$. Then with high probability there is no independent set $I$ in $G$ with both $|I \cap L| \geq \frac{(2+o(1))\log(d)}{d}n$
and $|I \cap R| \geq  \frac{(2+o(1))\log(d)}{d}n$.
\end{lemma}
\begin{proof}
For $k = \frac{(2+o(1))\log(d)}{d}n$, by the union bound, the probability there exists such a set is at most 
$$\binom{n}{k}^2 \left(1-\frac{d}{n}\right)^{k^2}$$ which is at most
$$\left(\left(\frac{en}{k}\right)^2\exp\left(-\frac{dk}{n}\right)\right)^k.$$ 
Substituting back the value of $k$, we get
$\left(\frac{e^2d^2}{4\log^2(d)}d^{-2}\right)^k<\left(e^2/16\right)^k \leq e^{-n^{0.9}}$, where the last inequality holds assuming $d \leq \frac{\log(n)}{100}.$
\end{proof}
Let us now prove \cref{lem:discrep}.
\begin{proof}[Proof of \cref{lem:discrep}]
    Set $q_n = (1-p_n)$ and for $t=0,\dots, {2n}$, define the process 
    $$M_t =  L_t-R_t + \frac{1}{2}\sum\limits_{s=0}^{t-1}\left(q_n^{R_{s}}- q_n^{L_{s}}\right).$$
    We claim that $M_t$ is a martingale with respect to the filtration generated by $(V_t, L_t, R_t)$. Indeed, it will be enough to show, 
    $$\EE\left[L_t - L_{t-1}-(R_t-R_{t-1})|L_{t-1},R_{t-1}, V_{t-1}\right] = \frac{1}{2}\left(q_n^{L_{t-1}}-q_n^{R_{t-1}}\right).$$
    Since $V_t$ and $L_t$ are independent
    $$\PP\left(L_t - L_{t-1}-(R_t-R_{t-1}) = 1|L_{t-1},R_{t-1}, V_{t-1}\right) = \PP\left(V_t = 1, L_t = L_{t-1} + 1\right) = \frac{1}{2}q_n^{R_{t-1}},$$
    and similarly
    $$\PP\left(L_t - L_{t-1}-(R_t-R_{t-1}) = -1|L_{t-1},R_{t-1}, V_{t-1}\right) = \frac{1}{2}q_n^{L_{t-1}}.$$
    which produces the desired result.
    It is also clear that $|M_t-M_{t-1}| < 2$ almost surely. So, by Azuma's inequality and a union bound, since $M_0 = 0,$
    \begin{equation} \label{eq:azumaapply}
        \PP\left(\max\limits_t|M_t| \geq n^{0.7}\right) = e^{-\Omega(n^{0.1})}.
    \end{equation}
    Our goal is to prove that $|L_{2n} - R_{2n}|$ is small. That will be achieved by dividing the time interval $[0,2n]$ into a constant (in $d$) number of time intervals, bounded by $\{t_i\}_{i>0}$, and showing that on each interval $[t_i,t_{i+1}]$ the auxiliary term in $M_t$, $\frac{1}{2}\sum\left(q_n^{R_{s}}- q_n^{L_{s}}\right)$ cannot increase by much. Under this condition \eqref{eq:azumaapply} will imply the same for $L_{t_i} - R_{t_i}$, and eventually for $L_{2n} - R_{2n}$.
    
    As suggested, we define a sequence of times by $t_0 = 0$ and $t_i = t_{i-1} + a_i\frac{n}{d}$, for $i > 0$, where 
    $$a_i = q_n^{(i-1)n^{0.8}}\frac{1-q_n^{n^{0.8}}}{1-q_n^{(i-1)n^{0.8}}}.$$
    A quick calculation reveals that,
    \begin{equation} \label{eq:cancelation}
        (1-q_n^{(i-1)n^{0.8}})a_i -q_n^{(i-1)n^{0.8}}(1-q_n^{n^{0.8}})=0.
    \end{equation}
    Moreover, the elementary inequalities $e^{\frac{x}{1-x}} \leq 1-x \leq e^{-x}$, which are valid for every $x <1$, show that for any fixed $i$, as long as $n$ is large enough,
    \begin{align} \label{eq:alower} 
        a_i \geq \frac{1}{2(i-1)}.
    \end{align}
    Thus, after taking an exponential in $d$ number of steps we will reach $2n$.
    
    To show that $L_t-R_t$ remains small, we define the events $A_i =\{\max\limits_{t\leq t_i} L_t-R_t < i\cdot n^{0.8}\}$. We would like to proceed by induction and estimate $\PP\left(A_{i+1}|A_i\right)$. Towards this, under the event $A_i$, suppose there exists $t_{i} \leq t < t_{i+1}$, such that,
    $L_t - R_t - (L_{t_i} - R_{t_i}) = \lfloor n^{0.8}\rfloor,$ and assume that $t$ is the minimal time that satisfies it.
    In this case, we shall show that $\frac{1}{2}|\sum\limits_{s=t_{i}}^{t-1}\left(q_n^{R_{s}}- q_n^{L_{s}}\right)|$ cannot be large. The idea is that for every $t_i \leq s \leq t$, under $A_i$ 
    $$
    |q_n^{R_{s}}- q_n^{L_{s}}| \leq 1 - q_n^{L_s - R_s} = 1 - q_n^{L_{t_i}-R_{t_i}}q_n^{L_s - R_s -(L_{t_i}-R_{t_i})} \leq 1 - q_n^{(i-1)n^{0.8}}q_n^{L_s - R_s -(L_{t_i}-R_{t_i})}.$$
    To maximize those terms under the constraint $L_s - R_s < n^{0.8}$, we would have exactly $\lfloor n^{0.8}\rfloor $ terms in which $L_s - R_s -(L_{t_i}-R_{t_i})$ increases linearly to $n^{0.8}$ and $t-n^{0.8}$ constant terms in which $L_s - R_s -(L_{t_i}-R_{t_i}) = n^{0.8}$. Thus, since $t<t_{i+1}$, and from the construction of $t_{i+1} -t_i$
      \begin{align*}
        \frac{1}{2}|\sum\limits_{s=t_{i}}^{t-1}\left(q_n^{R_{s}}- q_n^{L_{s}}\right)| &\leq \frac{1}{2}\left(t-t_{i}- q_n^{(i-1)n^{0.8}}\left[\sum\limits_{i=1}^{n^{0.8}}q_n^i + (t-t_{i}-n^{0.8})q_n^{n^{0.8}}\right]\right)\\
        &\leq \frac{1}{2}\left(t_{i+1}-t_{i}- q_n^{(i-1)n^{0.8}}\left[\sum\limits_{i=1}^{n^{0.8}}q_n^i + (t_{i+1}-t_{i}-n^{0.8})q_n^{n^{0.8}}\right]\right)\\
        &= \frac{n}{2d}\left(a_i - q_n^{(i-1)n^{0.8}}(1-q_n^{n^{0.8}}) - a_iq_n^{n^{0.8}}\right) + \frac{1}{2}q_n^{(i-1)n^{0.8}}q^{n^{0.8}}n^{0.8}\\
        & \leq \frac{1}{2}n^{0.8}.
    \end{align*}
    The equality uses the formula for a sum of a geometric series, along with the fact $\frac{1}{1-q_n} = \frac{n}{d}$. The last inequality uses \eqref{eq:cancelation} to eliminate the first term.
    This bound then implies,
    $$M_t - M_{t_i}= L_t - R_t - (L_{t_i}-R_{t_i}) + \frac{1}{2}\sum\limits_{s=t_i}^{t-1}\left(q_n^{R_{s}}- q_n^{L_{s}}\right) \geq n^{0.8} - \frac{1}{2}|\sum\limits_{s=t_{i}}^{t-1}\left(q_n^{R_{s}}- q_n^{L_{s}}\right)| \geq \frac{1}{2}n^{0.8}.$$
    But, as long as $n$ is large enough, this is impossible when
    $\max\limits_t |M_t| \leq n^{0.7}$, as in \eqref{eq:azumaapply}.
    We can thus conclude, for every $i > 0$
    \begin{equation} \label{eq:induction}
        \PP\left(A_{i+1}|A_i\right) \leq \PP\left(\max_t |M_t| > n^{0.7}\right) \leq e^{-\Omega(n^{0.1})}.
    \end{equation}
    Finally, from \eqref{eq:alower} we can choose $i_* = \Theta(2^d)$ to satisfy,
    $$t_{i_*} = \frac{n}{2d}\sum\limits_{i=1}^{t_*}a_i\geq\frac{n}{d}\log(i_*) > 2n. $$
    So, by repeating \eqref{eq:induction} for $\Theta(2^d)$ times, as long as $n$ is large enough, we can see the existence of a constant $c_d>0$, which depends only on $d$, such that
    $$\PP\left(L_{2n}-R_{2n}>n^{0.9}\right)\leq \PP\left(L_{2n}-R_{2n}>2^dn^{0.8}\right) \leq \PP(A_{i_*}) \leq e^{-c_dn^{0.1}},$$
    where we have used that $2^d \leq d^{\frac{1}{100}}$, since $d \leq \frac{1}{100} \log(n)$.
    By symmetry, $L_{2n} - R_{2n}$ and $R_{2n} - L_{2n}$ have the same distribution, so,
    $$\PP\left(|L_{2n}-R_{2n}|>n^{0.9}\right) \leq 2\PP\left(L_{2n}-R_{2n}>n^{0.9}\right),$$
    and the proof is finished.
\end{proof}

\section{Lower and Upper Bounds for the Time Complexity of Simulated Annealing in Trees}\label{sec:tree}

We now turn to study the performance of MP on trees and forests with an aim to prove \cref{thm:forest}. The proof of \cref{thm:forest} is separated into two parts. First, in \cref{sec:forsetlower} we construct a \emph{determinstic} hard instance for MP. The hard instance is a union of identical trees, each having weak hardness guarantees. It turns out that taking polynomially many copies of the same tree is enough to imply exponential lower bounds, and so proves the first point of \cref{thm:forest}. In \cref{sec:constantapproxtrees} we apply recent results about mixing times of MP on graphs with bounded treewidth to prove the second point of \cref{thm:forest}.
\subsection{Exponential lower bound} \label{sec:forsetlower}
We consider MP with an arbitrary non-adaptive fugacity schedule $\lambda_t$, and fix a small constant $\varepsilon \in (0,1/2)$.
To construct the hard instance, first consider a ``star-shaped'' tree $T_k$ that consists of a root $r$ connected to $k$ nodes 
$a_1, \ldots, a_k$, and each node $a_i$ has a single leaf neighbor $b_i$. That is, $T_k$ consists of a root connected to $k$ edge-disjoint length-$2$ paths. Let $A = \{a_1,\ldots,a_k\}$ and $B = \{b_1,\ldots,b_k\}$. The unique optimal independent set in $T_k$ is $r \cup B$, which has size $k+1$.

We first describe the first phase of the algorithm and show that it is hard for MP to include the root.
\begin{lemma}[Burn-in phase]
\label{lem:tree-burn-in}
The following holds with high ($1-o_k(1)$) probability. After $m=k^{1/2-\varepsilon}$ iterations, MP has at least $(1/2-\varepsilon)m$ vertices from $A$ (and, as a result, does not include the root $r$).
\end{lemma}
\begin{proof}
With high probability, the root $r$ is not selected during the first $m$ iterations. By a standard balls-and-bins argument~\cite{mitzenmacher2017probability}, with high probability the vertices selected during the first $m$ iterations are all distinct, and furthermore at most one vertex from each branch is selected ($a_i$ or $b_i$, but not both). Thus, MP adds all the vertices that are selected during the first $m$ iterations and does not delete any. With high probability, at least $(1/2-\varepsilon)m$ of the selected vertices belong to $A$.
\end{proof}

Now condition on the first $m$ steps of MP and suppose the high-probability event from \cref{lem:tree-burn-in} holds. Letting $I \subseteq [k]$ denote the set of branches $i$ for which MP includes $a_i$ at the end of $m$ steps, we have $|I| \ge (1/2-\varepsilon)m$. We will now focus on only the branches $I$ and show that MP continues to include at least one of the corresponding $a_i$'s for exponential time, blocking the root $r$ from being added.

For the analysis, it will be convenient to consider an auxiliary process MP', defined as follows. MP' has the same underlying graph $T_k$ and starts at the same state as MP at timestep $m$. MP' has the same update rule as MP except it \emph{never} adds the root $r$ (even if $r$ is selected and none of its neighbors are present). The random choices of the two processes are coupled so that MP and MP' share the same state until the first time that MP adds the root.

Fix a time horizon $T > m$. Our goal is to show that with high probability, at each timestep $t \le T$, MP includes at least one of the vertices $\{a_i \,:\, i \in I\}$. It suffices to show that the same holds for MP', as the presence of any $a_i$ prevents MP from adding the root. Therefore we turn our attention to the analysis of MP'.

Now reveal, and condition on, the choice of which \emph{branch} is selected by MP' at each timestep $t$ for $m < t \le T$. That is, we reveal the variable $\sigma_t$ which is equal to $i \in I$ if MP' chooses either $a_i$ or $b_i$ at step $t$, and equal to $\emptyset$ otherwise (i.e., if MP' selects $r$ or a branch outside $I$). The random choice between $a_i$ and $b_i$ is not revealed yet.

For $i \in I$, let $s_i$ denote the number of timesteps $t$ (where $m < t \le T$) for which $\sigma_t = i$. For each $i \in I$, we define an associated Markov chain $X^{(i)}_0,X^{(i)}_1,\ldots,X^{(i)}_{s_i}$ on the state space $\{\cA,\cB,\emptyset\}$ with initial state $X^{(i)}_0 = \cA$. The state $\cA$ encodes that MP' has $a_i$ (but not $b_i$), the state $\cB$ encodes that MP' has $b_i$ (but not $a_i$), and the state $\emptyset$ encodes that MP' has neither $a_i$ nor $b_i$. Every time branch $i$ updates (that is, $t$ such that $\sigma_t = i$), the Markov chain $X^{(i)}$ updates according to the following rules.
\begin{itemize}
    \item If the previous state is $\emptyset$, the new state is
    \[ \begin{cases} \cA & \text{with probability }1/2, \\ \cB & \text{with probability } 1/2. \end{cases} \]
    \item If the previous state is $\cA$, the new state is
    \[ \begin{cases} \emptyset & \text{w.p. }\,1/(2\lambda_t), \\ \cA & \text{w.p. }\, 1-1/(2\lambda_t). \end{cases} \]
    \item If the previous state is $\cB$, the new state is
    \[ \begin{cases} \emptyset & \text{w.p. }\,1/(2\lambda_t), \\ \cB & \text{w.p. }\, 1-1/(2\lambda_t). \end{cases} \]
\end{itemize}
Note that (conditioned on $\{\sigma_t\}$) the Markov chains $X^{(1)},\ldots,X^{(|I|)}$ are independent (since MP' never adds the root, by definition).

\begin{lemma}
For any fixed $i \in I$ and any fixed $0 \le \ell \le s_i$, we have $\PP(X^{(i)}_\ell = \cA) \ge 1/4$.
\end{lemma}
\begin{proof}
Proceed by induction on $\ell$, strengthening the induction hypothesis to include both (i) $\PP(X^{(i)}_\ell = \cA) \ge 1/4$ and (ii) $\PP(X^{(i)}_\ell = \cA) \ge \PP(X^{(i)}_\ell = \cB)$. The base case $\ell = 0$ is immediate, as $X^{(i)}$ is defined to start at state $\cA$. For the inductive step, we analyze the update rules given above. If $X^{(i)}_\ell$ takes values $\cA,\cB,\emptyset$ according to the vector of probabilities $(a,b,c)$, then $X^{(i)}_{\ell+1}$ is distributed according to the vector of probabilities
\[ (a',b',c') := \left(\frac{c}{2} + \left(1-\frac{1}{2\lambda_t}\right)a\,,\; \frac{c}{2} + \left(1-\frac{1}{2\lambda_t}\right)b\,,\; \frac{1}{2\lambda_t}(a+b) \right). \]
By induction we have $a \ge 1/4$ and $a \ge b$. Note that from $a \ge b$ we immediately have $a' \ge b'$, which proves (ii). For (i), $a \ge b$ implies $b \le 1/2$ and so
\[ a' = \frac{c}{2} + \left(1-\frac{1}{2\lambda_t}\right)a \ge \frac{c+a}{2} = \frac{1-b}{2} \ge \frac{1}{4}, \]
completing the proof.
\end{proof}

Therefore, by independence across branches, we have for any fixed timestep $t$ (with $m < t \le T$), the probability that MP' has \emph{none} of the vertices $\{a_i \,:\, i \in I\}$ is at most $(3/4)^{|I|}$. Taking a union bound over $t$, the probability that MP' intersects $\{a_i \,:\, i \in I\}$ until time $T$ is at least $1 - T (3/4)^{|I|}$. For $T \le \exp(k^{1/2-2\varepsilon})$ this is $1-o_k(1)$, recalling $|I| \ge (1/2-\varepsilon)m = (1/2-\varepsilon)k^{1/2-\varepsilon}$. As discussed above, we conclude the same result for the original process MP.

\begin{proposition} \label{prop:hardtrees}
    Fix any constant $\eta \in (0,1/2)$ and consider MP run on $T_k$ with an arbitrary non-adaptive fugacity schedule $\lambda_t$. With probability $1-o_k(1)$, MP does not add the root $r$ at any point during the first $\exp(k^\eta)$ iterations (and thus does not find a maximum independent set).
\end{proposition}
The hardness result in Theorem \ref{thm:forest} now follows by bootstrapping \cref{prop:hardtrees}.

\begin{proof}[Proof of first item in \cref{thm:forest}]
Let us first prove the result for a forest rather than a tree. Consider $k$ disjoint copies of $T_k$, which has a total of $k(2k+1)$ vertices. Additional isolated vertices can be added to create a forest of any desired size. To find the maximal independent set in the union, MP must add the root of every copy of $T_k$. Conditional on the exact number of updates per copy, MP evolves independently on each copy. Thus, the independence of the different copies and \cref{prop:hardtrees} yield an exponential bound $\exp(-\Omega(k))$ on the probability that it adds all the roots by time $\exp(k^\alpha)$.

We now modify the construction, turning the forest into a tree. Add a new vertex and connect it to the root of each $T_k$ (and also to any isolated vertices that were added to pad the instance size). The maximum independent set still includes the root of every $T_k$. If it were not for the new vertex, we have from above that with probability $1-\exp(-\Omega(k))$, there is at least one copy of $T_k$ where the root is never added before time $\exp(k^\alpha)$. With the new vertex present, this copy of $T_k$ still will never add the root before time $\exp(k^\alpha)$, as the new vertex can only affect $T_k$ by \emph{blocking} the root from being added. This completes the proof.
\end{proof}
\subsection{The Metropolis process can efficiently find approximate solutions on trees} \label{sec:constantapproxtrees}
Here we show that if one seeks an \emph{approximate} solution to the size of the largest independent set in a forest then MP finds such a solution efficiently: 

\begin{proof}[Proof of second item in \cref{thm:forest}]
Let $N$ be the maximum cardinality of an independent set in $F_n$. Since $F_n$ is bipartite we have that $N \geq n/2$. Set the fugacity as $\lambda \ge 4^{1/\varepsilon+\frac{\log_2 n}{\varepsilon n}}$. We upper bound the contribution of
independent sets whose size is smallar than $(1-\epsilon)N$ to the partition function:
$$\sum_{|I| < (1-\varepsilon)N} \lambda^{|I|} \leq 2^n \lambda^{(1-\varepsilon)N} \leq \lambda^N/n.$$
Therefore, for the stationary distribution of MP with this value of $\lambda$ the probability we get
an independent set of size smaller than $(1-\varepsilon) N$ is at most $1/n$. The desired result now follows from the fact~\cite{eppstein2021rapid,chen2023combinatorial} that MP on forests mixes in time $n^{O(1)}$ (assuming $\lambda$ is a fixed constant independent of $n$) as well as the fact~\cite{levin2017markov} that after $t_{\mathrm{mix}} \log n$
iterations of MP the total variation distance between the chain and the stationary distribution is at most $1/n$. Therefore, MP finds an independent set of size at least $(1-\varepsilon)N$ with probability at least $1-2/n$.
\end{proof}

\begin{remark}
    For graphs of treewidth $t$ it is known that the mixing time of MP is $n^{O(t)}$ ~\cite{eppstein2021rapid,chen2023combinatorial}. Since any graph with treewidth $t$ is $t$-degenerate (has a vertex of degree at most $t$) it has an independent set of size at least $n/(t+1)$. The argument above shows that assuming $t$ is constant (independent of $n$), MP will find a $1-\varepsilon$ approximation to the optimal solution in polynomial time with high probability. 
\end{remark}

\section{Conclusion}

We have proved super polynomial lower bounds on the time complexity of SA and the MP when approximating the size of a maximum independent set in several graph families. 
Many questions remain. Can we strengthen our lower bounds in the sparse case to hold for SA and not only for MP with a fixed temperature?
What are the approximation guarantees these heuristics can achieve efficiently for the independent set problem? Our techniques show MP cannot achieve efficiently an approximation ratio of $\omega(\log^3 n/n)$. We believe that MP will achieve worse approximation ratios compared to state-of-the-art methods~\cite{feige2004approximating} that achieve roughly $\Theta(\log^3 n/n)$
approximation (up to low-order multiplicative terms). Studying this question further is an interesting future direction.

Studying the performance of SA in approximating the independent set problem for specific graphs families could yield new findings about the strengths and weaknesses of these algorithms. For example, for planar graphs (such as trees) we proved that SA will not find the optimal solution in polynomial time. 
While NP-hard to solve exactly,
the independent set problem admits a polynomial-time approximation scheme in planar graphs~\cite{baker1994approximation}.
How well does SA approximate $\alpha(G)$ in planar graphs? As another example, we do not know whether SA
(or MP) finds an optimal independent set in random bipartite graphs in polynomial time where each edge between the sides occurs independently with probability $p$. Efficient recovery of optimal independent sets in such graphs could extend to additional graph families such as bipartite expanders. If SA fails to find the optimal solutions in polynomial time in random bipartite graphs that could lead to a simpler construction of hard bipartite instances.

It would be interesting to study the efficacy of SA in approximating other NP-hard problems such as vertex cover and min bisection. Finally, one could examine whether the lower bounds established here for MP when initialized from the empty set could lead to lower bounds on the mixing time of the Glauber dynamics for the hard-core model when initialized from the empty set.

\bibliographystyle{alpha}
\bibliography{reference1}
\appendix
\section{Metropolis vs.~Greedy Algorithms}\label{sec:greedyvsmetro}
Below we sketch the existence of instances to the maximum independent set problem where the Metropolis process can far exceed (in polynomial time) the performance of greedy algorithms.

\subsection{Degree-based greedy}
Recall that the degree-based greedy algorithm for the maximum independent set problem iteratively chooses a vertex $v$ of minimal degree (breaking ties arbitrarily),
adds it to an independent set (initialized to be the empty set), and deletes $v$ and all its neighbors from the input graph. The algorithm terminates once all vertices
have been deleted. 

Consider the following family of graphs $G_n$. It consists of an independent set $I$ of size $n$, a clique $C$ disjoint from $I$ of size $n$, and an additional vertex $r$ (overall $2n+1$ vertices).
$I$ and $C$ are connected to one another via a complete bipartite graph: every vertex in $I$ is connected to every vertex in $C$. 
Finally, $r$ is connected to all vertices of $I$ (and no other vertex). The maximum independent set in this graph is of size $n$.

The greedy algorithm will initially choose $r$, as it has a minimal degree, add it to the solution, and delete $r$ and its neighbors from $G_n$. Thereafter it will choose 
a single vertex from $C$. 

For Metropolis, set $\lambda = n^2$. Suppose the algorithm chooses $r$ and an additional vertex from $C$. The expected number of steps it takes until both vertices are deleted is polynomial (and the expected time a single vertex not from $I$ is deleted is even smaller and polynomial as well). It follows that the expected time until Metropolis picks a vertex $v \in I$ is polynomial (and hence by Markov with high probability a vertex from $I$ will be picked within polynomially many steps).

After a vertex $u\in I$ is chosen we have, by applying a coupon collector argument, that the expected time until all vertices from $I$ are added is $n(\log(n) +O(1))$, and furthermore with probability at least $1-1/n$ all vertices from $I$ are chosen after $2n \log(n)$ iterations~\cite{mitzenmacher2017probability}. By the choice of $\lambda$, with high probability, no deletions are going to occur during these $2n \log(n)$ steps. Hence, with high probability, MP will find an optimal independent set in $G_n$ in polynomial time. Hence MP will efficiently find an independent set of size $n$ with probability $1-o(1)$ whereas the greedy algorithm will find an independent set of size $2$.  

\subsection{Randomized greedy}

In randomized greedy (RG) vertices are added at random to a growing independent set, initialized as the empty set. In every iteration, a vertex $v$ is chosen uniformly at random from the remaining graph (initialized to be the input graph)
and added to the set. Thereafter $v$ and all its neighbors are deleted from the graph. The algorithm terminates once all vertices have been deleted. 

Consider the following graph family $H_n$. Fix $\varepsilon \in (0,1)$ to be detemined later. There are $n$ disjoint copies of an independent set of size $n^{\varepsilon}$ connected with a complete bipartite graph to a clique of size $n$. Hence the graph has $n(n+n^{\varepsilon})$ vertices. The maximum size of an independent set in $H_n$ is $n^{1+\varepsilon}.$

For each copy, the probability a vertex from the independent set is chosen by randomized greedy is 
$$\frac{n^{\varepsilon}}{n+n^{\varepsilon}} \leq \frac{1}{n^{1-\varepsilon}}.$$
If a vertex from an independent set is chosen by RG all other vertices from the independent set are chosen as well, leading to an expected size of at most $\frac{1}{n^{1-2\varepsilon}}.$ If a clique vertex is chosen the contribution is a single vertex. Hence the expected size RG for a single copy is at most $1+\frac{1}{n^{1-2\varepsilon}}$. The expected size of the solution found by RG for the entire graph is at most $n+\frac{1}{n^{2\varepsilon}}.$ By Chernoff, the probability RG will return an independent set of size larger than $2(n+n^{2\varepsilon})$ 
is at most $\exp(-\Omega(n).)$

We now discuss how the Metropolis process fairs with $H_n$. Set $\lambda=n^2$. We begin by considering a single copy of an independent set $I$ connected to a clique $C$. Even if a clique vertex is chosen the expected time until a vertex from $I$ is chosen (which requires the deletion of the clique vertex) is polynomial. It follows (by Markov) that with high probability after a polynomial number of iterations, a vertex from $I$ is chosen. 

Next, as the deletion probability is $1/n^2$ it follows by a coupon collector argument~\cite{mitzenmacher2017probability} that all vertices of $I$ are chosen with probability at least $1-1/n^2$ after $3 n^{\varepsilon} \log(n)$ iterations once the first vertex from $I$ is chosen. 

We now need to bound the probability the process deletes many vertices from $I$ once all 
of the vertices of $I$ are chosen. Toward this end we need the following two results (see for example ~\cite{levin2017markov}).

\begin{lemma}
For a reversible Markov chain with transition matrix $P$ over a finite state space $\Gamma$ it holds that for every $x,y \in \Gamma$ and integer $t \geq 1$:
$$\pi(x)P^t(x,y)=\pi(y)P^t(y,x),$$
where $P^t$ is the probability of moving from $x$ to $y$ in $t$ steps and $\pi$ is the stationary distribution of $P$ which exists since the chain is reversible. 
\end{lemma}

\begin{definition}
Consider the following birth and death Markov chain $B$ on the state space $\{0,\ldots k\}$: for $0\leq i< k$ the probability of moving from $i$ to $i+1$ is $p_i$; the probability of moving from $0 < i \leq k$ to $i-1$ is $q_i$ and the probability the chain stays at $i$ is $r_i=1-p_i-q_i$. For $j>0$ let $w_j=\prod_{i=1}^j \frac{p_{i-1}}{q_i}$ and set $w_0=1$. 
\end{definition}
\begin{lemma}
The stationary distribution $\pi$ of the birth and death process $B$ satisfies for $j \in [k]:$ 
$$\pi(j)=w_j/\sum_{i=0}^{k}w_i.$$
\end{lemma}

\begin{corollary}
Let $0\leq a<b \leq k$. For the chain $B$ above the probability we move in $t$ steps from $b$ to $a$ is at most $\left(\prod_{j=a+1}^b\frac{p_{i-1}}{q_i}\right)^{-1}.$
\end{corollary}
\begin{proof}
The probability we move from $b$ to $a$ in $t$ steps is $P^t(b,a)$ where $P$ is the transition matrix of $B$. By the lemma above
$P^t(b,a) \leq \pi(a)/\pi(b)=\left(\prod_{j=a+1}^b\frac{p_{i-1}}{q_i}\right)^{-1}$.
\end{proof}

We can now prove:
\begin{theorem}
Suppose we run Metropolis on $H_n$ with $\lambda=n^2$. Fix arbitrary $\delta \in (0,1)$.
With high probability the algorithm will find an independent set of size $(1-\delta)n^{1+\varepsilon}.$
\end{theorem}
\begin{proof}
By the discussion above and using a union bound over the $n$ copies, after a polynomial number of steps $T$, MP will pick all $n^{\varepsilon}$ vertices in each of the $n$ disjoint independent sets in $H_n$ at least once. The probability of adding a vertex to an independent set that is not fully occupied is at least $n^{-\varepsilon} \geq 1/n$ (conditioned on the vertex from the independent set chosen by the algorithm) and the probability of deletion is at most $1/n^2$. We have that for each copy, once all $n^{\varepsilon}$ vertices of an independent set are chosen, the probability more than $\delta n^{\varepsilon}$ vertices are deleted in $T$ steps is at most $n^{-\delta n^{\varepsilon}-2}$. Therefore in $T$ iterations, with high probability the algorithms maintains an independent set whose intersection with each copy is of size at least $(1-\delta)n^{\varepsilon}.$ Therefore the algorithm finds an independent set of size at least $n(1-\delta)\cdot n^\varepsilon$ in polynomial time with high probability.      
\end{proof}

Plugging in $\varepsilon=1/2$ we get:
\begin{corollary}
   There exists a graph with $m$ vertices such that randomized greedy finds with high probability an independent set of size $O(m^{1/2})$, whereas MP finds with high probability in polynomial time an independent set of size $\Omega(m^{3/4})$.
\end{corollary}
\end{document}